\documentclass[11pt,english]{article}
\usepackage[T1]{fontenc}
\usepackage[latin9]{inputenc}
\usepackage{geometry}
\usepackage{float}
\usepackage{url}
\usepackage{amsthm}
\usepackage{amsmath}
\usepackage{graphicx}
\usepackage{amssymb}

\floatstyle{ruled}
\newfloat{algorithm}{tbp}{loa}
\floatname{algorithm}{Algorithm}

\theoremstyle{plain}
\theoremstyle{plain}
\newtheorem{thm}{Theorem}
  \theoremstyle{definition}
  \newtheorem{defn}[thm]{Definition}
  \theoremstyle{remark}
  \newtheorem*{rem*}{Remark}
  \theoremstyle{plain}
  \newtheorem{lem}[thm]{Lemma}

\usepackage{babel}

\begin{document}

\title{A Combinatorial Study of Linear Deterministic Relay Networks}

\date{S. M. Sadegh Tabatabaei Yazdi and Serap A. Savari%
\thanks{The authors are with the Department of Electrical and Computer Engineering,
Texas A\&M University, College Station, TX 77843 USA. Their e-mail
addresses are \protect\url{sadegh@neo.tamu.edu} and \protect\url{savari@ece.tamu.edu}.%
}}
\maketitle
\begin{abstract}
In the last few years the so--called {}``linear deterministic''
model of relay channels has gained popularity as a means of studying
the flow of information over wireless communication networks, and
this approach generalizes the model of wireline networks which is
standard in network optimization. There is recent work extending the
celebrated max--flow/min--cut theorem to the capacity of a unicast
session over a linear deterministic relay network which is modeled
by a layered directed graph. This result was first proved by a random
coding scheme over large blocks of transmitted signals. We demonstrate
the same result with a simple, deterministic, polynomial--time algorithm
which takes as input a single transmitted signal instead of a long
block of signals. Our capacity-achieving transmission scheme for a
two--layer network requires the extension of a one--dimensional Rado--Hall
transversal theorem on the independent subsets of rows of a row--partitioned
matrix into a two--dimensional variation for block matrices. To generalize
our approach to larger networks we use the submodularity of the capacity
of a cut for our model and show that our complete transmission scheme
can be obtained by solving a linear program over the intersection
of two polymatroids. We prove that our transmission scheme can achieve
the max-flow/min-cut capacity by applying a theorem of Edmonds about
such linear programs. We use standard submodular function minimization
techniques as part of our polynomial--time algorithm to construct
our capacity-achieving transmission scheme.
\end{abstract}

\section{Introduction}

Network information theory \cite[Ch. 15]{key-7} attempts to model
aspects of large communication networks such as interference, cooperation,
and noise that are often overlooked in network optimization theory.
Relay channels \cite[\S 15.7]{key-7} are an example of a network
information problem in which there is a source, a unique destination,
and at least one intermediary transmitter--receiver pair which is
instrumental to the communication between the source and the destination.
In this paper we focus on a simplified mathematical model for the
\emph{wireless} relay channel which we describe via a directed graph
$\mathcal{N}(\mathcal{V},\mathcal{E})$, where $\mathcal{V}$ denotes
the set consisting of the source node, the destination node, and all
relay nodes. Here the source node only sends signals, the destination
node only receives signals, and a relay node can both receive signals
and transmit any function of its incoming messages. The transmission
of signals in this network obeys two rules: 1-- any signal which is
sent by node $i$ is broadcast to every node $j$ such that $ij\in\mathcal{E},$
and 2-- the signal which is received by node $j$ is a linear combination
of all signals that are broadcast to it and an independent additive
noise signal which is typically modeled by a Gaussian random variable.
Finding the capacity of a wireless relay channel has long been a challenging
and important open problem. Avestimehr, Diggavi, and Tse \cite{key-1,key-2}
recently proposed a simplified \emph{linear deterministic relay network
}model in which the effects of broadcasting, interference, and noise
are captured by linear transformations of the transmitted signals.
One motivation for the study of this model is the result in \cite{key-1,key-2}
that the capacity of any wireless relay channel with Gaussian noise
is within a constant additive factor of the capacity of a corresponding
linear deterministic model. More recent work \cite{key-15,key-16}
has connected the linear deterministic model to the approximate capacity
of other relay channels and to the design of near--optimal coding
schemes for them. We next summarize the model and some of the results
of \cite{key-1,key-2}: 

The authors of \cite{key-1,key-2} focus on \emph{layered} directed
graphs $\mathcal{N}$ with set of nodes $\mathcal{V}=\mathcal{O}_{1}\cup\mathcal{O}_{2}\cup\cdots\cup\mathcal{O}_{M}.$
Here $\mathcal{O}_{1}=\left\{ \mathcal{O}_{1}(1)\right\} =\left\{ \mathcal{S}\right\} $
and $\mathcal{O}_{M}=\left\{ \mathcal{O}_{M}(1)\right\} =\left\{ \mathcal{D}\right\} $
respectively denote the source and destination nodes, and $\mathcal{O}_{i}=\left\{ \mathcal{O}_{i}(1),\cdots,\mathcal{O}_{i}(m_{i})\right\} $
denotes the set of relay nodes in the $i$th layer for $i\in\left\{ 2,\cdots,M-1\right\} $.
Each edge in the graph is from some node in $\mathcal{O}_{i},i\in\left\{ 1,\cdots,M-1\right\} $,
to a node in $\mathcal{O}_{i+1}$. Observe that the study of an arbitrary
directed network can also be placed into this framework if one instead
studies its time--expanded representation \cite{key-1,key-2}. In
this case $\mathcal{O}_{i}$ corresponds to the possible behaviors
of the network from symbol time $(i-1)\tau$ to symbol time $i\tau-1$
for some positive integer $\tau$. 

In the layered network, every node $\mathcal{O}_{i}(j),i\in\left\{ 2,\cdots,M-1\right\} ,$
receives a column vector $\mathbf{y}_{i}^{t}(j)$ at time $t$ and
transmits a column vector $\mathbf{x}_{i}^{t}(j)$ at time $t.$ Node
$\mathcal{S}$ transmits vector $\mathbf{x}_{\mathcal{S}}^{t}=\mathbf{x}_{1}^{t}$
and node $\mathcal{D}$ receives vector $\mathbf{y}_{\mathcal{D}}^{t}=\mathbf{y}_{M}^{t}.$
The elements of all transmitted and received vectors belong to some
fixed finite field $\mathbb{F}_{q}.$ Every edge $\mathcal{O}_{i}(j)\mathcal{O}_{i+1}(k)\in\mathcal{E}$
represents a communication channel which is described by a matrix
$G_{i}(j,k)$ with entries from $\mathbb{F}_{q}$ called the \emph{transfer
function} from $\mathbf{x}_{i}^{t}(k)$ to vector $\mathbf{y}_{i+1}^{t}(j)$.
Define the block matrix $G_{i}$ by $G_{i}=\left[G_{i}(j,k)\right]$
for $j\in\left\{ 1,\cdots,m_{i+1}\right\} =\left[m_{i+1}\right]$
and $k\in\left[m_{i}\right].$ Let $\mathbf{x}_{i}^{t}$ be the transmitted
vector from layer $i$ and $\mathbf{y}_{i+1}^{t}$ be the received
vector at layer $i+1$; i.e., \[
\mathbf{x}_{i}^{t}=\left[\begin{array}{c}
\mathbf{x}_{i}^{t}(1)\\
\vdots\\
\mathbf{x}_{i}^{t}(m_{i})\end{array}\right]\mbox{and}\quad\mathbf{y}_{i+1}^{t}=\left[\begin{array}{c}
\mathbf{y}_{i+1}^{t}(1)\\
\vdots\\
\mathbf{y}_{i+1}^{t}(m_{i+1})\end{array}\right].\]
 The communication channel from layer $i$ to layer $i+1$ is characterized
by the following relationship:\begin{equation}
\mathbf{y}_{i+1}^{t}=G_{i}\cdot\mathbf{x}_{i}^{t}.\label{eq:linear1}\end{equation}
 Assume the communication session begins at time 1 and ends at time
$\tau$ and has desired rate of transmission $R;$ i.e., node $\mathcal{S}$
wishes to send message $\omega$ which is chosen randomly from a set
$\left\{ 1,\cdots,q^{\tau R}\right\} $ of messages to node $\mathcal{D}$
at the end of the session. The communication protocol proceeds from
layer to layer. Node $\mathcal{O}_{i}(j)$ transmits vectors $\mathbf{x}_{i}^{1}(j),\cdots,\mathbf{x}_{i}^{\tau}(j)$
to nodes in layer $\mathcal{O}_{i+1},$ and node $\mathcal{O}_{i+1}(k)$
transmits signals to the nodes in next layer after receiving vectors
$\mathbf{y}_{i+1}^{1}(k),\cdots,\mathbf{y}_{i+1}^{\tau}(k).$ At every
time instant $t$, vector $\mathbf{x}_{\mathcal{S}}^{t}$ is some
function of $\omega$ and $\mathbf{x}_{i}^{t}(j)$ is some function
of $\mathbf{y}_{i}^{1}(j),\cdots,\mathbf{y}_{i}^{\tau}(j).$ There
are two natural questions about this model: First, what is the \emph{capacity}
or maximum rate of information in this network? Second, among capacity--achieving
schemes, how can one optimize the duration $\tau$ and the complexity
of the relay functions used? 

The first of these questions was initially addressed in \cite{key-1,key-2}.
To study the capacity $\mathcal{C}$ of the network, we first define
a \emph{cut} $\Omega$ as a subset of the nodes $\mathcal{V}.$ A
cut separates $\mathcal{S}$ from $\mathcal{D}$ if $\mathcal{S}\in\Omega$
and $\mathcal{D}\in\bar{\Omega}=\mathcal{V}\backslash\Omega.$ The
\emph{transfer function} of the cut $\Omega$, $G(\Omega),$ is defined
as a block diagonal matrix with $(M-1)\times(M-1)$ blocks. The $i^{th}$
diagonal block, $G_{i}(\Omega)$ is the submatrix of $G_{i}$ consisting
of the transfer functions from the transmitted vectors of the nodes
in $\Omega\cap\mathcal{O}_{i}$ to the received vectors of the nodes
in $\bar{\Omega}\cap\mathcal{O}_{i+1}$ for $i\in\left[M-1\right],$
and each off--diagonal block is an all--zero matrix. Avestimehr, Diggavi
and Tse used a min--cut upper bound on the rate of transmissions \cite[Thm. 15.10.1]{key-7}
and a random coding argument to show 
\begin{thm}
\label{thm:channel-capacity} For any cut $\Omega$, $\mathcal{C}(\Omega)=\mbox{rank}(G(\Omega))=\sum_{i=1}^{M-1}\mbox{rank}(G_{i}(\Omega)).$
Furthermore, the capacity of network $\mathcal{N}$ as defined above
is $\mathcal{C}=\min_{\Omega\mbox{ separates }\mathcal{S}\mbox{ and }\mathcal{D}}\mathcal{C}(\Omega).$
\end{thm}
The achievability argument in \cite{key-1,key-2} is based on a linear
scheme in which source node $\mathcal{S}$ initially encodes the message
$\omega$ as a vector in $\mathbb{F}_{q}^{\tau R}$ denoted by $\mathbf{y}_{\mathcal{S}}(\omega)$.
Node $\mathcal{S}$ and relay node $\mathcal{O}_{i}(j)$ respectively
generate transmitted signals by the linear transformations\[
\left[\begin{array}{c}
\mathbf{x}_{\mathcal{S}}^{1}\\
\vdots\\
\mathbf{x}_{\mathcal{S}}^{\tau}\end{array}\right]=F_{\mathcal{S}}\cdot\mathbf{y}_{\mathcal{S}}(\omega)\quad\mbox{and}\quad\left[\begin{array}{c}
\mathbf{x}_{i}^{1}(j)\\
\vdots\\
\mathbf{x}_{i}^{\tau}(j)\end{array}\right]=F_{i}(j)\cdot\left[\begin{array}{c}
\mathbf{y}_{i}^{1}(j)\\
\vdots\\
\mathbf{y}_{i}^{\tau}(j)\end{array}\right].\]
 It is shown in \cite{key-1,key-2} that if the encoding matrices
$F_{\mathcal{S}}$ and $F_{i}(j)$ are chosen randomly with a uniform
distribution over the space of all matrices over the field $\mathbb{F}_{q}$,
if $\tau$ is sufficiently large, and if $R\leq\mathcal{C}$, then
the destination node $\mathcal{D}$ will, with probability approaching
1, receive $\tau R$ linearly independent linear combinations of the
message vector $\mathbf{y}_{\mathcal{S}}(\omega)$ from which it will
be able to decode message $\omega.$ 

Since the complexity of the transmission scheme in \cite{key-1,key-2}
is large and increases with $\tau$, we seek a deterministic, low--complexity
transmission scheme that is capacity--achieving and processes only
one signal $\mathbf{x}_{i}^{t}(j)$ at a time for each $i$ and $j$,
i.e., $\tau=1.$ We remark that \cite{key-3} considered similar issues
for transmissions over a binary field. We will discuss the approach
of \cite{key-3} in Section 1.2.

\subsection{Our Results and Techniques}

Our algorithm has two main steps. First we propose an algorithm to
transmit signals from layer $\mathcal{O}_{i}$, $i\in\left\{ 1,\cdots,M-1\right\} ,$
to layer $\mathcal{O}_{i+1}$ in an optimal way. In the second step
we extend our algorithm to the full network and prove that it is capacity--achieving.
Since our transmission scheme manipulates only one signal $\mathbf{x}_{i}^{t}(j)$
at a time for each $i$ and $j$, we hereafter drop the time superscript. 

We define a \emph{flow} of the block matrix $G_{i}$ as follows:
\begin{defn}
Let $\mathbf{d}_{i}=\left(\ell_{i}(1),\cdots,\ell_{i}(m_{i});\ell_{i+1}(1),\cdots,\ell_{i+1}(m_{i+1})\right)$
be a vector of non--negative integers that satisfies $\sum_{j=1}^{m_{i}}\ell_{i}(j)=\sum_{j=1}^{m_{i+1}}\ell_{i+1}(j)\doteq R_{\mathbf{d}_{i}}.$
We say that matrix $G_{i}$ \emph{supports} flow $\mathbf{d}_{i}$
if there exists a full rank $R_{\mathbf{d}_{i}}\times R_{\mathbf{d}_{i}}$
submatrix $G_{\mathbf{d}_{i}}$ of $G_{i}$ such that $G_{\mathbf{d}_{i}}$
is an intersection of $\ell_{i}(j)$ columns of the $j$th column
block of $G_{i},$ $j\in\left[m_{i}\right],$ with $\ell_{i+1}(k)$
rows of the $k$th row block of $G_{i}$, $k\in\left[m_{i+1}\right]$.
(See Figure \ref{fig:example}.(a).) We further say that such a submatrix
$G_{\mathbf{d}_{i}}$ is a \emph{solution} for flow $\mathbf{d}_{i}$.
\end{defn}
For the physical interpretation of flow, suppose matrix $G_{i}$ supports
flow $\mathbf{d}_{i}.$ Consider the subvector $\mathbf{x}_{\mathbf{d}_{i}}$
of $\mathbf{x}_{i}$ and the subvector $\mathbf{y}_{\mathbf{d}_{i+1}}$
of $\mathbf{y}_{i+1}$ which correspond to the transfer matrix $G_{\mathbf{d}_{i}}$.
Furthermore, let $\mathbf{x}_{\mathbf{d}_{i}}(j)$ and $\mathbf{y}_{\mathbf{d}_{i+1}}(k)$
respectively denote the parts of these subvectors that belong to vector
$\mathbf{x}_{i}(j)$ and $\mathbf{y}_{i+1}(k)$. If the entries of
$\mathbf{x}_{i}$ which are not part of $\mathbf{x}_{\mathbf{d}_{i}}$
are set to zero, then $\mathbf{y}_{\mathbf{d}_{i+1}}$ will uniquely
determine $\mathbf{x}_{\mathbf{d}_{i}}$ since $G_{\mathbf{d}_{i}}$
is a full rank matrix. Hence $R_{\mathbf{d}_{i}}$ units of information
flow from the nodes in $\mathcal{O}_{i}$ to the nodes in $\mathcal{O}_{i+1}$
during a transmission. We next introduce the notion of the flow from
$\mathcal{O}_{1}$ to $\mathcal{O}_{M}$ supported by network $\mathcal{N}.$
For convenience we will consider a more general network $\mathcal{N}$
with an arbitrary number of nodes in the first and last layers as
opposed to the single node each at the first and last layers of relay
channel models.
\begin{defn}
\label{def:matrix-flow} Suppose non--negative integers $\ell_{1}(1),\cdots,\ell_{1}(m_{1}),$
and $\ell_{M}(1),\cdots,\ell_{M}(m_{M})$ satisfy $\sum_{j=1}^{m_{1}}\ell_{1}(j)=\sum_{j=1}^{m_{M}}\ell_{M}(j)\doteq R$.
We say that vector $\mathbf{d}=\left(\ell_{1}(1),\cdots,\ell_{1}(m_{1});\ell_{M}(1),\cdots,\ell_{M}(m_{M})\right)$
is a rate--$R$ flow supported by network $\mathcal{N}$ if for every
$i\in\left\{ 2,\cdots,M-1\right\} $ there exists non--negative integers
$\ell_{i}(1),\cdots,\ell_{i}(m_{i})$ such that vector $\mathbf{d}_{j}=\left(\ell_{j}(1),\cdots,\ell_{j}(m_{j});\ell_{j+1}(1),\cdots,\ell_{j+1}(m_{j+1})\right)$,
$j\in\left[M-1\right],$ is a rate--$R$ flow supported by matrix
$G_{j}$. (See Figure \ref{fig:example}.(b).)
\end{defn}
Every flow for network $\mathcal{N}$ is determined by the submatrices
$G_{\mathbf{d}_{i}}$ and the corresponding row and column indices
of $G_{i}$. Let us return to the case where $\mathcal{N}$ has a
single node each in the first and last layers. Suppose that network
$\mathcal{N}$ supports a rate--$R$ flow $\mathbf{d}=\left(R;R\right).$
Then given $G_{\mathbf{d}_{i}},$ $i\in\left[M-1\right],$ a simple
coding scheme that achieves rate $R$ can be defined as follows:%
\begin{figure}
\centering \label{Flo:example}\includegraphics[bb=0bp 0bp 595bp 220bp,clip,scale=0.7]{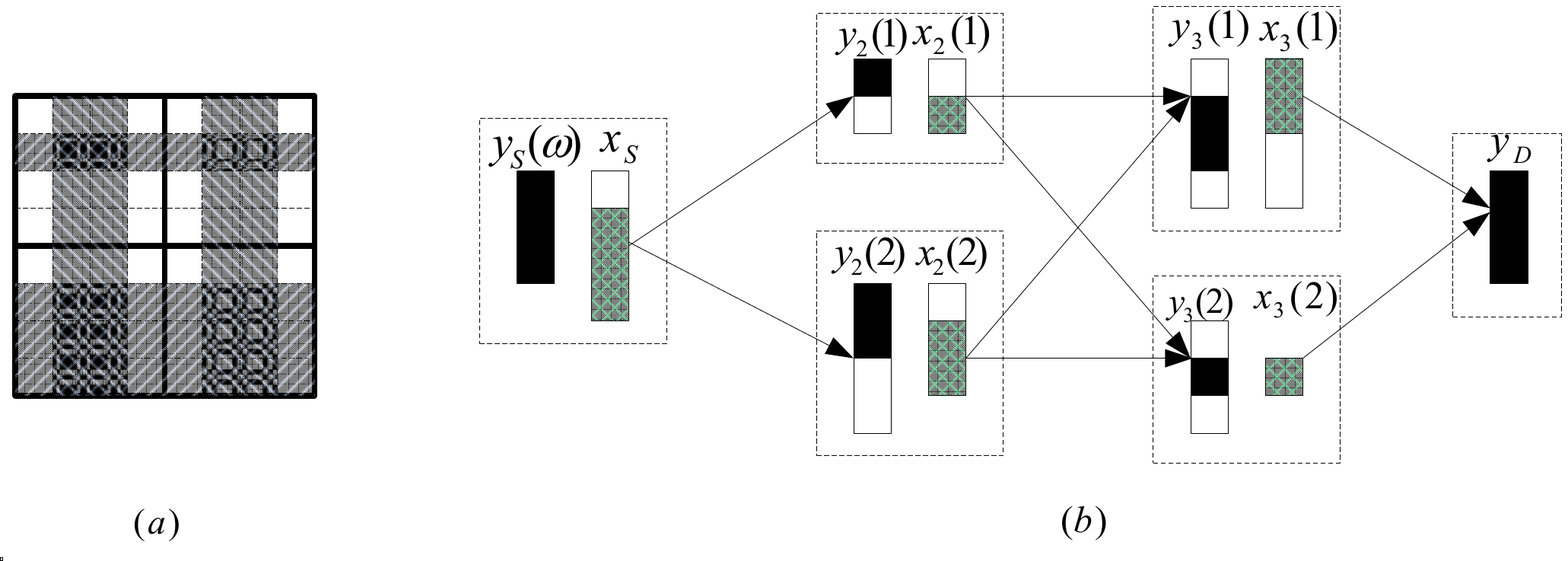}

\caption{\label{fig:example}(a) An example of a matrix flow for vector $\mathbf{d}_{i}=(2,2;1,3)$
in a matrix $G_{i}$ with four blocks. Each small square is an entry
of $G_{i}$ and matrix $G_{\mathbf{d}_{i}}$ is the intersection of
dashed rows with dashed columns. (b) An example of network flow $\mathbf{d}=(3;3)$
in a network with four layers. The solid part of vector $\mathbf{y}_{i}(j)$
denotes the entries of $\mathbf{y}_{\mathbf{d}_{i}}(j)$ and the dashed
part of vector $\mathbf{x}_{i}(j)$ denotes the entries of $\mathbf{x}_{\mathbf{d}_{i}}(j)$.
The flow vectors are $\mathbf{d}_{1}=(3;1,2)$, $\mathbf{d}_{2}=(1,2;2,1)$,
and $\mathbf{d}_{3}=(2,1;3).$}

\end{figure}

\subsubsection*{Transmission Scheme: }

Given the length--$R$ encoded vector $\mathbf{y}_{\mathcal{S}}(\omega)$,
node $\mathcal{S}$ generates vector $\mathbf{x}_{\mathcal{S}}=\mathbf{x}_{1}$
by setting $\mathbf{x}_{\mathbf{d}_{1}}$ to the vector $\mathbf{y}_{\mathcal{S}}(\omega)$
and by setting the other entries of $\mathbf{x}_{\mathcal{S}}$ to
zero. The transformation at every relay node $\mathcal{O}_{i}(j)$
is similar: after receiving vector $\mathbf{y}_{i}(j)$, node $\mathcal{O}_{i}(j)$
extracts the subvector $\mathbf{y}_{\mathbf{d}_{i}}(j)$ with length
$\ell_{i}(j)$ and sets $\mathbf{x}_{\mathbf{d}_{i}}(j)=\mathbf{y}_{\mathbf{d}_{i}}(j)$.
The remaining entries of $\mathbf{x}_{i}(j)$ are set to zero. Finally
node $\mathcal{D}$ first decodes subvector $\mathbf{y}_{\mathbf{d}_{M}}$
from the received vector $\mathbf{y}_{\mathcal{D}}=\mathbf{y}_{M}$
and then extracts the encoded message $\mathbf{y}_{\mathcal{S}}(\omega)$.
Observe that for every $i$, $\mathbf{x}_{\mathbf{d}_{i}}=G_{\mathbf{d}_{i}}^{-1}\cdot\mathbf{y}_{\mathbf{d}_{i+1}}$
and $\mathbf{x}_{\mathbf{d}_{i}}=\mathbf{y}_{\mathbf{d}_{i}}.$ These
imply that $\mathbf{y}_{\mathcal{S}}(\omega)=G_{\mathbf{d}_{1}}^{-1}G_{\mathbf{d}_{2}}^{-1}\cdots G_{\mathbf{d}_{M-1}}^{-1}\mathbf{y}_{\mathbf{d}_{M}}.$
Since the matrices $G_{\mathbf{d}_{i}}$ are nonsingular, the decoding
operation is well defined. 

Our main technical result in Section 2 is the following theorem providing
the necessary and sufficient conditions for matrix $G_{i}$ to support
a flow $\mathbf{d}_{i}$: 
\begin{thm}
\label{thm:matrix-transversal} For any subsets $U\subseteq\left[m_{i}\right]$
and $W\subseteq\left[m_{i+1}\right]$ define the block matrix $G_{i}(W,U)$
as the intersection of the row blocks of $G_{i}$ having indices in
$W$ with the column blocks of $G_{i}$ having indices in $U.$ Matrix
$G_{i}$ supports a flow $\mathbf{d}_{i}$ if and only if for all
$U\subseteq\left[m_{i}\right]$ and for all $W\subseteq\left[m_{i+1}\right]$ 

\begin{equation}
\mbox{rank}(G_{i}(W,U))\geq\sum_{j\in U}\ell_{i}(j)+\sum_{k\in W}\ell_{i+1}(k)-R_{\mathbf{d}_{i}}.\label{eq:transversal-condition}\end{equation}

\end{thm}
This combinatorial property of matrices is, to our knowledge, the
first two--dimensional result of this type and may be of independent
interest in the theory of matrices. Theorem \ref{thm:matrix-transversal}
holds for matrices with entries from an arbitrary field and is therefore
more general than its application for this relay problem. We prove
the necessity of Theorem \ref{thm:matrix-transversal} by examining
the relationships of the ranks of various submatrices of block matrix
$G_{\mathbf{d}_{i}}$.

Our sufficiency argument is more technical and involves a divide--and--conquer
procedure and an inductive argument to prove the existence of $G_{\mathbf{d}_{i}}$.
The basis of our inductive argument will be the Rado--Hall transversal
theorem \cite[Ch. 7]{key-5}, which states the necessary and sufficient
conditions for the existence of independent structures in a collection
of subsets of elements of a matroid. In the special case of matrices
we have:
\begin{thm}
\label{thm:(Rado's-Theorem)}(The Rado--Hall Theorem) Let $H$ be
a matrix with column blocks labeled $1,\cdots,m$. For $U\subseteq\left[m\right]$
let $H(U)$ denote the submatrix of $H$ which is formed by the column
blocks with indices in $U.$ Given non--negative integers $\ell_{1},\cdots,\ell_{m}$,
there are $\ell_{1}+\cdots+\ell_{m}$ linearly independent columns
of $H$ with exactly $\ell_{i}$ of the columns from $H(\left\{ i\right\} ),$
$i\in\left[m\right],$ if and only if for every $U$ we have $\mbox{rank}(H(U))\geq\sum_{i\in U}\ell_{i}.$
\end{thm}
In the study of matrices Theorem \ref{thm:matrix-transversal} can
be viewed as an extension of the Rado--Hall Theorem. The Rado--Hall
Theorem has many variations for matroids with different constraints
on the cardinalities of the independent sets \cite{key-6}. However,
the rank function of a matroid is, by definition, both submodular
and nondecreasing. The columns of a matrix form a ground set for a
matroid, and the rank function of a subset of columns is defined as
the dimension of the subspace spanned by these columns; this is also
true if the columns are replaced by the rows of the matrix. We next
define a natural extension of the rank function on the union of all
rows and columns of a matrix which is motivated by the cuts of a network.
Let $\mathcal{P}$ and $\mathcal{Q}$ respectively denote the set
of rows and columns of the matrix. Define the rank function of a subset
$T\subseteq\mathcal{P}\cup\mathcal{Q}$ where $T=P\cup Q$ with $P\subseteq\mathcal{P}$
and $Q\subseteq\mathcal{Q},$ to be the dimension of the spanning
subspace of the submatrix which is formed by the intersection of the
rows in $P$ and the columns which are \emph{not} in $Q.$ In Lemma
\ref{lem:9} we study this rank function for the block matrices of
network $\mathcal{N}$ and prove that it is submodular. In that discussion
the nodes of the two consecutive layers of the network respectively
represent the column and row blocks of a matrix. We comment that the
rank function arising from the study of cuts in our model has an important
difference from the rank function in earlier transversal theorems.
Observe that since the rank of both sets $\emptyset$ and $\mathcal{P}\cup\mathcal{Q}$
is zero, the rank function we introduce is not monotone. Korte and
Lovász \cite{key-24} initiated the study of a generalization of matroids
known as greedoids with rank functions which are monotone but not
necessarily submodular. They have also derived a transversal theory
of greedoids similar to the Rado--Hall Theorem \cite{key-24}. Theorem
\ref{thm:matrix-transversal} appears to be the first extension of
the Rado--Hall Theorem in which the rank function is submodular but
not monotone. This is interesting because monotonicity plays a central
role in the proofs of the previous results. 

In Section 3, we prove the following extension of Theorem \ref{thm:matrix-transversal}:
\begin{thm}
\label{thm:Network flow} A network $\mathcal{N}$ with capacity $\mathcal{C}$
supports flow $\mathbf{d}=\left(R;R\right)$ if and only if $R\leq\mathcal{C}.$
\end{thm}
Edmonds and Giles \cite{key-17} introduced a generalization of classical
network flow known as \emph{submodular flow,} where the classical
flow conservation constraints are replaced by submodular flow constraints
on certain subsets of nodes. We prove a submodularity property of
the cut function in Lemma \ref{lem:90}, and it is possible to show
that our notion of flow for network $\mathcal{N}$ is a special case
of submodular flow. The Edmonds and Giles theory of submodular flow
does not directly imply variations of the max--flow min--cut theorems.
We therefore study our flow using an earlier result of Edmonds \cite{key-19}.
We prove in Section 3 that the maximum flow in our setting is a linear
programming optimization over the intersection of two polymatroids.
We apply a corollary of the polymatroid intersection theorem \cite{key-19}
to show that the maximum rate of a flow is the capacity $\mathcal{C}$
of the network and that the corresponding flow can be achieved by
non--negative integer vectors $\mathbf{d}_{i}$. 

In the appendix, we demonstrate an algorithm to construct a capacity--achieving
code for network $\mathcal{N}$ which is strongly polynomial time
in the size of the graph and in the size of the matrices $G_{i}$.
In first step of the analysis we show that: 
\begin{thm}
\label{thm:alg-matrix}Given matrix $G_{i}$ and an achievable flow
vector $\mathbf{d}_{i},$ the submatrix $G_{\mathbf{d}_{i}}$ can
be computed in polynomial time.
\end{thm}
We can in principle use the divide--and--conquer argument for the
proof of Theorem \ref{thm:matrix-transversal} to obtain a recursive
algorithm for this problem, but since the analysis is difficult we
take a different approach. We will demonstrate that testing the conditions
of Theorem \ref{thm:matrix-transversal} for matrix $G_{i}$ and flow
$\mathbf{d}_{i}$ is equivalent to showing that a certain submodular
function has a non--negative minimum. It is well known (see, e.g.,
\cite{key-11,key-18}) that there are polynomial--time algorithms
to find the minimum of a submodular function. Our algorithm to construct
$G_{\mathbf{d}_{i}}$ checks which rows of $G_{i}$ can be removed
without violating the conditions of Theorem \ref{thm:matrix-transversal}
and then removes them one by one. The same procedure is next repeated
for the columns. The final part of this step is to establish that
the remaining matrix is a valid choice for $G_{\mathbf{d}_{i}}$. 

The second step establishes the following result. 
\begin{thm}
\label{thm:alg-network}The capacity $\mathcal{C}$ of a deterministic
relay network $\mathcal{N}$ can be computed in polynomial time. Given
the network flow vector $\mathbf{d}=(R;R)$ with $R\leq\mathcal{C},$
the flow vectors $\mathbf{d}_{i}$ for every matrix $G_{i}$, $i\in\left\{ 1,\cdots,M-1\right\} ,$
can be computed in polynomial time.
\end{thm}
We prove this theorem in the appendix by applying the algorithm in
\cite{key-12} for optimizing a linear function over the intersection
of two polymatroids. Theorems \ref{thm:alg-matrix} and \ref{thm:alg-network}
together imply a polynomial--time algorithm for finding a transmission
scheme for network $\mathcal{N}.$

\subsection{Related Work}

Earlier work \cite{key-13}, \cite{key-14} obtained capacity results
for a different type of deterministic relay network in which the nodes
broadcast data but the signals are received without interference.
The paper \cite{key-3} considers the same problem we address here,
but restricts $\mathbb{F}_{q}$ to a binary field. The approach of
\cite{key-3} is based on a path augmentation argument similar to
the Ford-Fulkerson algorithm (see, e.g., \cite{key-22}) and involves
a new network in which every node $\mathcal{O}_{i}(j)$ is replaced
by several nodes with each corresponding to a different entry of vector
$\mathbf{x}_{i}(j)$ or $\mathbf{y}_{i}(j).$ In the new network there
is an edge between a pair of nodes if the analogous entry in the transfer
function $G_{i}(k,j)$ is equal to one. For an edge $e$, we respectively
denote its tail and head by $x(e)$ and $y(e).$ Edges $e_{1},\cdots,e_{K}$
are said to be \emph{independent} if the transfer function from the
vector $\left(x(e_{1}),\cdots,x(e_{K})\right)$ to the vector $\left(y(e_{1}),\cdots,y(e_{K})\right)$
has full rank. The scheme in \cite{key-3} finds $\mathcal{C}$ disjoint
paths in the new network, where every cut that separates $\mathcal{S}$
from $\mathcal{D}$ intersects $\mathcal{C}$ independent edges of
these paths. There does not appear to be a natural way to extend the
approach of \cite{key-3} to arbitrary finite fields. We use a different
procedure to construct the full rank submatrices needed for our transmission
scheme.
\begin{rem*}
The missing proofs in the body of the paper can be found in the appendix.
\end{rem*}

\section{Proof of Theorem \ref{thm:matrix-transversal} }

At several steps of our proof we remove rows or columns from matrix
$G_{i}$ to find a suitable submatrix $G_{\mathbf{d}_{i}}.$ Unless
otherwise stated, assume that each such intermediate submatrix $T$
of $G_{i}$ maintains the same partition of row blocks and column
blocks as original matrix $G_{i}.$ In other words, each row (column)
of $T$ consists of a subset of the entries of some row (column) of
$G_{i}$, and the two rows (columns) have the same row (column) block
index in their respective matrices. For $U\subseteq\left[m_{i}\right]$
and $W\subseteq\left[m_{i+1}\right]$, let the block submatrix $T(W,U)$
denote the intersection of the row blocks of $T$ having indices in
$W$ with the column blocks of $T$ having indices in $U.$ Notice
that some row and/or column blocks of $T$ may be empty.

\subsection{Proof of Necessity}

For $U\subseteq\left[m_{i}\right]$ and $W\subseteq\left[m_{i+1}\right]$,
$G_{\mathbf{d}_{i}}(W,U)$ is a submatrix of $G_{i}(W,U)$. Therefore,
$\mbox{rank}(G_{i}(W,U))\geq\mbox{rank}(G_{\mathbf{d}_{i}}(W,U)).$
From the submodularity of the rank function we have:

\begin{equation}
\mbox{rank}(G_{\mathbf{d}_{i}}(\left[m_{i+1}\right],U))\leq\mbox{rank}(G_{\mathbf{d}_{i}}(W,U))+\mbox{rank}(G_{\mathbf{d}_{i}}(\left[m_{i+1}\right]\backslash W,U)).\label{eq:nece1}\end{equation}
 Since $G_{\mathbf{d}_{i}}$ is a full rank square matrix all of its
columns are independent and hence \begin{equation}
\mbox{rank}(G_{\mathbf{d}_{i}}(\left[m_{i+1}\right],U))=\sum_{j\in U}\ell_{i}(j).\label{eq:nece2}\end{equation}
 By the monotonicity of the rank function \begin{equation}
\mbox{rank}(G_{\mathbf{d}_{i}}(\left[m_{i+1}\right]\backslash W,U))\leq\mbox{rank}(G_{\mathbf{d}_{i}}(\left[m_{i+1}\right]\backslash W,\left[m_{i}\right])).\label{eq:nece3}\end{equation}
Since all of the rows of $G_{\mathbf{d}_{i}}$ are independent we
have\begin{equation}
\mbox{rank}(G_{\mathbf{d}_{i}}(\left[m_{i+1}\right]\backslash W,\left[m_{i}\right]))=\sum_{k\in\left[m_{i+1}\right]\backslash W}\ell_{i+1}(k)=R_{\mathbf{d}_{i}}-\sum_{k\in W}\ell_{i+1}(k).\label{eq:nece4}\end{equation}
 The relations (\ref{eq:nece1})--(\ref{eq:nece4}) imply the necessity
of the condition\[
\mbox{rank}(G_{i}(W,U))\geq\sum_{j\in U}\ell_{i}(j)+\sum_{k\in W}\ell_{i+1}(k)-R_{\mathbf{d}_{i}}.\]

\subsection{Proof of Sufficiency}

Assume throughout this subsection that the conditions of Theorem \ref{thm:matrix-transversal}
are satisfied. First suppose that $m_{i+1}=1$ and $W=\left\{ 1\right\} .$Then
for every set $U\subseteq\left[m_{i}\right]$ the inequality (\ref{eq:transversal-condition})
reduces to:\[
\mbox{rank}(G_{i}(\left\{ 1\right\} ,U))\geq\sum_{j\in U}\ell_{i}(j)+\ell_{i+1}(1)-R_{\mathbf{d}_{i}}.\]
 Since the definition of vector $\mathbf{d}_{i}$ implies $\ell_{i+1}(1)=R_{\mathbf{d}_{i}},$
it follows that $\mbox{rank}(G_{i}(\left\{ 1\right\} ,U))\geq\sum_{j\in U}\ell_{i}(j).$
By Theorem \ref{thm:(Rado's-Theorem)}, there exists a rank--$R_{\mathbf{d}_{i}}$
submatrix $\tilde{G_{i}}$ of matrix $G_{i}$ which consists a subset
of $\ell_{i}(j)$ columns, for $j\in\left[m_{i}\right],$ from each
column block $G_{i}(\left\{ 1\right\} ,\left\{ j\right\} ).$ Since
$\mbox{rank}(\tilde{G_{i}})=R_{\mathbf{d}_{i}},$ $\tilde{G_{i}}$
has a submatrix $G_{\mathbf{d}_{i}}$ consisting of $R_{\mathbf{d}_{i}}$
independent rows. $G_{\mathbf{d}_{i}}$ is a solution for flow $\mathbf{d}_{i}.$ 

We can similarly argue the existence of a solution $G_{\mathbf{d}_{i}}$
for flow $\mathbf{d}_{i}$ when $m_{i}=1.$ Next suppose $m_{i}\geq2$
and $m_{i+1}\geq2.$ We will use induction and a divide--and--conquer
procedure to prove the sufficiency result. For our inductive hypothesis
we assume that Theorem \ref{thm:matrix-transversal} is true for any
block matrix $G_{i}$ consisting of $n_{i+1}\times n_{i}$ blocks
where $n_{i}\leq m_{i},$ $n_{i+1}\leq m_{i+1},$ and $m_{i+1}\times m_{i}\neq n_{i+1}\times n_{i}$. 

Let $U\subseteq\left[m_{i}\right]$ and $W\subseteq\left[m_{i+1}\right]$.
We say that $G_{i}(W,U)$ is a \emph{tight} submatrix of $G_{i}$
if the inequality (\ref{eq:transversal-condition}) holds with equality
for $U$ and $W$. 
\begin{lem}
\label{lem:tight-submatrix}Either $G_{i}$ has a tight submatrix
or $G_{i}$ has a submatrix $T$ for which for all $\hat{U}\subseteq\left[m_{i}\right]$
and $\hat{W}\subseteq\left[m_{i+1}\right]$, \begin{equation}
\mbox{rank}(T(\hat{W},\hat{U}))\geq\sum_{j\in\hat{U}}\ell_{i}(j)+\sum_{k\in\hat{W}}\ell_{i+1}(k)-R_{\mathbf{d}_{i}}\label{eq:6A}\end{equation}
 and $T$ has a tight submatrix; i.e., (\ref{eq:6A}) hold with equality
for some $\tilde{U}\subseteq\left[m_{i}\right]$ and $\tilde{W}\subseteq\left[m_{i+1}\right].$
In the latter case we replace $G_{i}$ with $T$ for the rest of the
proof of Theorem \ref{thm:matrix-transversal}.
\end{lem}
By our previous argument, $G_{i}$ has one or more tight submatrices.
We call a tight submatrix $G_{i}(W,U)$ a \emph{proper} submatrix
if it is \emph{not} one of the following cases:
\begin{enumerate}
\item $|U|=m_{i}$ and $|W|=1$, or
\item $|U|=1$ and $|W|=m_{i+1}$. 
\end{enumerate}
For the rest of the proof of Theorem \ref{thm:matrix-transversal}
we need a proper submatrix. We have the following result:
\begin{lem}
\label{lem:proper-submatrix} Either (a) $G_{i}$ has a proper submatrix
or (b) it has a submatrix $T$ which satisfies (\ref{eq:6A}) for
all $\hat{U}\subseteq\left[m_{i}\right]$ and $\hat{W}\subseteq\left[m_{i+1}\right]$
and $T$ has a proper submatrix $T(W,U)$. If case (b) applies then
we replace $G_{i}$ with the corresponding submatrix $T$ for the
rest of the proof of Theorem \ref{thm:matrix-transversal}.
\end{lem}
Let $P=G_{i}(W,U)$ be a proper submatrix of $G_{i}.$ Next we reorder
and relabel the row blocks and the column blocks of $G_{i}$ such
that $P=G_{i}(\left[n_{i+1}\right],\left[n_{i}\right]).$ It is straightforward
to reverse the ordering and relabeling operation at the end of our
argument. We label the four parts of matrix $G_{i}$ as the following,
where $A,B,$ and/or $L$ may possibly be empty matrices: \[
G_{i}=\left[\begin{array}{cc}
P & A\\
B & L\end{array}\right].\]
We denote by $G_{A}$ the matrix $\left[\begin{array}{cc}
P & A\end{array}\right].$ We further consider a partition of $G_{A}$ into blocks as the following:\[
G_{A}=\left[\begin{array}{c|c|c|c}
G_{i}(1,\left[n_{i}\right]) & G_{i}(1,n_{i}+1) & \cdots & G_{i}(1,m_{i})\\
\hline \vdots & \vdots & \ddots & \vdots\\
\hline G_{i}(n_{i+1},\left[n_{i}\right]) & G_{i}(n_{i+1},n_{i}+1) & \cdots & G_{i}(n_{i+1},m_{i})\end{array}\right].\]
 Next consider the vector \[
\mathbf{d}_{A}=\left(\mbox{rank}(P),\ell_{i}(n_{i}+1),\ell_{i}(n_{i}+2),\cdots,\ell_{i}(m_{i});\ell_{i+1}(1),\cdots,\ell_{i+1}(n_{i+1})\right).\]
 We verify that $\mathbf{d}_{A}$ is a valid flow vector. By the tightness
of matrix $P$ we have:\[
\mbox{rank}(P)=\sum_{j=1}^{n_{i+1}}\ell_{i+1}(j)+\sum_{k=1}^{n_{i}}\ell_{i}(k)-R_{\mathbf{d}_{i}}=\sum_{j=1}^{n_{i+1}}\ell_{i+1}(j)-\sum_{k=n_{i}+1}^{m_{i}}\ell_{i}(k).\]
Therefore we have\begin{equation}
\mbox{rank}(P)+\sum_{j=n_{i}+1}^{m_{i}}\ell_{i}(j)=\sum_{k=1}^{n_{i+1}}\ell_{i+1}(k)\doteq R_{A}.\label{eq:rank(P)}\end{equation}

\begin{lem}
\label{lem:3}Matrix $G_{A}$ supports flow $\mathbf{d}_{A}.$ 
\end{lem}
Let $G_{\mathbf{d}_{A}}=\left[\begin{array}{cc}
P_{t} & A_{t}\end{array}\right]$ be the submatrix of $G_{A}$ corresponding to a solution for flow
$\mathbf{d}_{A}$; here $P_{t}$ is a submatrix of $P$ and $A_{t}$
is a submatrix of $A.$ We let $G'_{\mathbf{d}_{A}}=\left[\begin{array}{cc}
P_{c} & A_{t}\end{array}\right],$ where $P_{c}$ is the submatrix of $P$ consisting of the rows of
$P$ that are used for the construction of $P_{t};$ $P_{t}$ is then
a submatrix of $P_{c}$ consisting of a subset of its columns. Notice
that matrices $P_{t}$ and $P_{c}$ have a natural partition into
$n_{i+1}\times n_{i}$ block matrices which is induced by the block
structure of matrix $P.$ Next let $G_{B}=\left[\begin{array}{c}
P_{c}\\
B\end{array}\right]$ and partition $G_{B}$ into blocks as follows:\[
G_{B}=\left[\begin{array}{c|c|c}
P_{c}(\left[n_{i+1}\right],1) & \cdots & P_{c}(\left[n_{i+1}\right],n_{i})\\
\hline G_{i}(n_{i+1}+1,1) & \cdots & G_{i}(n_{i+1}+1,n_{i})\\
\hline \vdots & \ddots & \vdots\\
\hline G_{i}(m_{i+1},1) & \cdots & G_{i}(m_{i+1},n_{i})\end{array}\right].\]
Next we define the vector\[
\mathbf{d}_{B}=\left(\ell_{i}(1),\cdots,\ell_{i}(n_{i});\mbox{rank}(P),\ell_{i+1}(n_{i+1}+1),\cdots,\ell_{i+1}(m_{i+1})\right).\]
 Since\[
\mbox{rank}(P)=\sum_{j=1}^{n_{i+1}}\ell_{i+1}(j)+\sum_{k=1}^{n_{i}}\ell_{i}(k)-R_{\mathbf{d}_{i}}=\sum_{k=1}^{n_{i}}\ell_{i}(k)-\sum_{j=n_{i+1}+1}^{m_{i+1}}\ell_{i+1}(j),\]
 it follows that\[
\sum_{k=1}^{n_{i}}\ell_{i}(k)=\mbox{rank}(P)+\sum_{j=n_{i+1}+1}^{m_{i+1}}\ell_{i+1}(j)\doteq R_{B}\]
 and $\mathbf{d}_{B}$ is a valid flow vector. Furthermore we have:
\begin{lem}
\label{lem:F_B flow}Matrix $G_{B}$ supports flow $\mathbf{d}_{B}.$ 
\end{lem}
Let $G_{\mathbf{d}_{B}}=\left[\begin{array}{c}
P_{ct}\\
B_{t}\end{array}\right]$ be the submatrix of $G_{B}$ corrsponding to a solution for flow
$\mathbf{d}_{B}$; here $P_{ct}$ is a submatrix of $P_{c}$ and $B_{t}$
is a submatrix of $B.$ We let $G'_{\mathbf{d}_{B}}=\left[\begin{array}{c}
P_{cr}\\
B_{t}\end{array}\right],$ where $P_{cr}$ is the submatrix of $P_{c}$ consisting of the columns
of $P_{c}$ that are used for the construction of $P_{ct};$ $P_{ct}$
is then a submatrix of $P_{cr}$ consisting of a subset of its rows.
Finally, let $L_{t}$ be the submatrix of $L$ obtained by intersecting
the set of columns with indices matching those of the columns of $A_{t}$
with the set of rows with indices matching those of $B_{t}.$ Observe
that\[
G_{\mathbf{d}_{i}}=\left[\begin{array}{cc}
P_{cr} & A_{t}\\
B_{t} & L_{t}\end{array}\right]\]
 is a submatrix of $G_{i}.$ Our final step is the following lemma:
\begin{lem}
\label{lem:F_i flow}Matrix $G_{\mathbf{d}_{i}}$ as defined above
is a solution for flow $\mathbf{d}_{i}$ for matrix $G_{i}.$ 
\end{lem}
To summarize the preceding argument, we have established the existence
of a solution $G_{\mathbf{d}_{i}}$ when $m_{i}=1$, when $m_{i+1}=1,$
and when $m_{i}\geq2$ and $m_{i+1}\geq2.$ Our proof of Theorem \ref{thm:matrix-transversal}
is complete.

\section{Proof of Theorem \ref{thm:Network flow}}

We prove a more general statement. Consider a network $\mathcal{N}$
with an arbitrary number of nodes in the first and last layers. For
a rate--$R$ flow vector $\mathbf{d}=\left(\ell_{1}(1),\cdots\ell_{1}(m_{1});\ell_{M}(1),\cdots,\ell_{M}(m_{M})\right)$
we show that:
\begin{thm}
\label{thm:extension}Network $\mathcal{N}$ supports rate--$R$ flow
$\mathbf{d}$ if and only if for every cut $\Omega,$ \begin{equation}
\mathcal{C}(\Omega)\geq\sum_{\mathcal{O}_{1}(j)\in\Omega}\ell_{1}(j)+\sum_{\mathcal{O}_{M}(k)\in\bar{\Omega}}\ell_{M}(k)-R.\label{eq:extension-condition}\end{equation}

\end{thm}
Notice that for $M=2,$ Theorem \ref{thm:extension} reduces to Theorem
\ref{thm:matrix-transversal}. Also if $\mathcal{N}$ has a single
node each in the first and last layers, then $\ell_{1}(1)=\ell_{M}(1)=R,$
and hence (\ref{eq:extension-condition}) and Theorem \ref{thm:channel-capacity}
imply that for every cut $\Omega$ with $\mathcal{S}\in\Omega$ and
$\mathcal{D}\in\bar{\Omega}$, $\mathcal{C}(\Omega)\geq\mathcal{C}\geq R.$
Thus, Theorem \ref{thm:extension} implies Theorem \ref{thm:Network flow}.

\subsection{Proof of Theorem \ref{thm:extension}}

We use induction on $M.$ For $M=2,$ Theorem \ref{thm:matrix-transversal}
gives the result. For $M>2,$ choose $K\in\left\{ 2,\cdots,M-1\right\} .$
Define networks $\mathcal{N}_{A}$ and $\mathcal{N}_{B}$ to respectively
be the subnetworks of $\mathcal{N}$ with node set $\mathcal{O}_{1}\cup\cdots\cup\mathcal{O}_{K}$
and $\mathcal{O}_{K}\cup\cdots\cup\mathcal{O}_{M}$. The next step
of our argument is to show that the inductive hypothesis and (\ref{eq:extension-condition})
imply the existence of non-negative integers $\ell_{K}(1),\cdots,\ell_{K}(m_{K})$
such that $\sum_{i=1}^{m_{K}}\ell_{K}(i)=R$ and networks $\mathcal{N}_{A}$
and $\mathcal{N}_{B}$ support the rate--$R$ flows\begin{align*}
 & \mathbf{d}_{A}=(\ell_{1}(1),\cdots,\ell_{1}(m_{1});\ell_{K}(1),\cdots,\ell_{K}(m_{K}))\\
 & \mathbf{d}_{B}=(\ell_{K}(1),\cdots,\ell_{K}(m_{K});\ell_{M}(1),\cdots,\ell_{M}(m_{M})).\end{align*}
 This step would establish that $\mathcal{N}$ supports flow $\mathbf{d}$
since submatrix $R_{\mathbf{d}_{i}}$ can be obtained from the solution
to $\mathcal{N}_{A}$ for $i\in\left\{ 1,\cdots,K-1\right\} $ and
from the solution to $\mathcal{N}_{B}$ for $i\in\left\{ K,\cdots,M-1\right\} .$ 

By the inductive hypothesis if the desired $\ell_{K}(1),\cdots,\ell_{K}(m_{K})$
exist then they are non--negative integers which form a feasible solution
to the following system of linear constraints:

\begin{equation}
\begin{cases}
\mathcal{C}(\Omega_{A})\geq\sum_{\mathcal{O}_{1}(j)\in\Omega_{A}}\ell_{1}(j)+\sum_{\mathcal{O}_{K}(k)\in\bar{\Omega}_{A}}\ell_{K}(k)-R, & \mbox{for every cut }\Omega_{A}\mbox{ in }\mathcal{N}_{A},\\
\mathcal{C}(\Omega_{B})\geq\sum_{\mathcal{O}_{K}(j)\in\Omega_{B}}\ell_{K}(j)+\sum_{\mathcal{O}_{M}(k)\in\bar{\Omega}_{B}}\ell_{M}(k)-R, & \mbox{for every cut }\Omega_{B}\mbox{ in }\mathcal{N}_{B},\mbox{ and }\\
\sum_{i=1}^{m_{K}}\ell_{K}(i)=R.\end{cases}\label{eq:17}\end{equation}
 For any set $T\subseteq\mathcal{O}_{K}$ define:\begin{align}
 & f_{A}(T)=\min\left\{ \mathcal{C}(\Omega_{A})-\sum_{\mathcal{O}_{1}(j)\in\Omega_{A}}\ell_{1}(j)+R:\mathcal{O}_{K}\cap\bar{\Omega}_{A}=T\right\} \label{eq:f_A}\\
 & f_{B}(T)=\min\left\{ \mathcal{C}(\Omega_{B})-\sum_{\mathcal{O}_{M}(j)\in\bar{\Omega}_{B}}\ell_{M}(j)+R:\mathcal{O}_{K}\cap\Omega_{B}=T\right\} .\label{eq:f_B}\end{align}
 Then the system (\ref{eq:17}) of linear constraints is equivalent
to\begin{equation}
\begin{cases}
\sum_{\mathcal{O}_{K}(j)\in T}\ell_{K}(j)\leq f_{A}(T), & \mbox{for every }T\subseteq\mathcal{O}_{K}\\
\sum_{\mathcal{O}_{K}(j)\in T}\ell_{K}(j)\leq f_{B}(T), & \mbox{for every }T\subseteq\mathcal{O}_{K}\\
\sum_{i=1}^{m_{K}}\ell_{K}(i)=R.\end{cases}\label{eq:lp}\end{equation}

\begin{lem}
\label{lem:7}The functions $f_{A}(T)$ and $f_{B}(T)$ are I) submodular,
II) nondecreasing, and satisfy III) $f_{A}(\emptyset)=0$ and $f_{B}(\emptyset)=0.$
Notice that function $f$ is submodular if for every $T_{1}$ and
$T_{2}$: \[
f(T_{1})+f(T_{2})\geq f(T_{1}\cap T_{2})+f(T_{1}\cup T_{2})\]
 and is nondecreasing if for every $T_{1}\subseteq T_{2},$ $f(T_{1})\leq f(T_{2}).$
\end{lem}
Referring to terminology in polyhedral optimizations (see \cite[\S 5.15]{key-9})
a polytope $P$ is integer if and only if each vertex of $P$ has
integral coordinates. If a polyhedron $P=\left\{ \mathbf{x}:A\mathbf{x}\leq\mathbf{b}\right\} $
in $n$ dimensions is integer, then any linear programming problem
$\max\left\{ \mathbf{c}^{T}\mathbf{x}:A\mathbf{x}\leq\mathbf{b}\right\} $
with a finite solution must have a solution with integral coordinates. 

Let $f$ be a submodular function on some set $V$ with $v$ elements.
The \emph{polymatroid} associated with $f$ is: \[
P_{f}=\left\{ \mathbf{x}\in\mathbb{R}^{v}:\mathbf{x}\geq\mathbf{0},\quad x(U)\leq f(U)\quad\mbox{for every }U\subseteq V\right\} ,\]
where we define $\mathbf{x}=\left[\begin{array}{ccc}
x(1) & \cdots & x(v)\end{array}\right]^{T}$ and $x(U)=\sum_{u\in U}x(u).$ 
\begin{thm}
\label{thm:12}(\cite{key-19}) Let $f_{1}$ and $f_{2}$ be nondecreasing
submodular set functions with integer values on $V$ with $f_{1}(\emptyset)=f_{2}(\emptyset)=0.$
Then $P_{f_{1}}\cap P_{f_{2}}$ is integer and for each $U\subseteq V,$
\begin{equation}
\max\left\{ x(U):\mathbf{x}\in P_{f_{1}}\cap P_{f_{2}}\right\} =\min_{T\subseteq U}\left(f_{1}(T)+f_{2}(U\backslash T)\right).\label{eq:34}\end{equation}

\end{thm}
For the submodular set functions $f_{A}$ and $f_{B},$ define the
polymatroids: \begin{align*}
P_{f_{A}}=\left\{ \mathbf{x}\in\mathbb{R}^{m_{K}}:\mathbf{x}\geq\mathbf{0},\quad x(T)\leq f_{A}(T)\quad\mbox{for every }T\subseteq\mathcal{O}_{K}\right\} ,\\
P_{f_{B}}=\left\{ \mathbf{x}\in\mathbb{R}^{m_{K}}:\mathbf{x}\geq\mathbf{0},\quad x(T)\leq f_{B}(T)\quad\mbox{for every }T\subseteq\mathcal{O}_{K}\right\} .\end{align*}
 For (\ref{eq:lp}) to have a non--negative and integral solution,
$\max\left\{ x\left(\mathcal{O}_{K}\right):\mathbf{x}\in P_{f_{A}}\cap P_{f_{B}}\right\} \geq R$
is clearly necessary. To show sufficiency suppose $\mathbf{y}\in P_{f_{A}}\cap P_{f_{B}}$
achieves $\max\left\{ x\left(\mathcal{O}_{K}\right):\mathbf{x}\in P_{f_{A}}\cap P_{f_{B}}\right\} \geq R$.
Then for every choice of $0\leq\ell_{K}(j)\leq y(j),$ $\left[\begin{array}{ccc}
\ell_{K}(1) & \cdots & \ell_{K}(m_{K})\end{array}\right]^{T}\in P_{f_{A}}\cap P_{f_{B}}$ and so we choose $\ell_{K}(j)$ such that $\sum_{i=1}^{m_{K}}\ell_{K}(i)=R.$
Lemma \ref{lem:7} and Theorem \ref{thm:12} imply that:\begin{equation}
\max\left\{ x\left(\mathcal{O}_{K}\right):\mathbf{x}\in P_{f_{A}}\cap P_{f_{B}}\right\} =\min_{T\subseteq\mathcal{O}_{K}}\left(f_{A}\left(T\right)+f_{B}\left(\mathcal{O}_{K}\backslash T\right)\right)\label{eq:35}\end{equation}
and the optimum can be achieved by a non--negative integer solution.
Theorem \ref{thm:extension} follows from (\ref{eq:35}) and the following
lemma:
\begin{lem}
\textup{\label{lem:final-lemma}$\min_{T\subseteq\mathcal{O}_{K}}\left(f_{A}(T)+f_{B}(\mathcal{O}_{K}\backslash T)\right)\geq R$
if and only if for every cut $\Omega$ in $\mathcal{N},$ \begin{equation}
\mathcal{C}(\Omega)\geq\sum_{\mathcal{O}_{1}(j)\in\Omega}\ell_{1}(j)+\sum_{\mathcal{O}_{M}(k)\in\bar{\Omega}}\ell_{M}(k)-R.\label{eq:necessity}\end{equation}
}\end{lem}

\section*{Appendix}

\subsection*{Proof of Theorem \ref{thm:alg-matrix}:}

\begin{algorithm}
\caption{\label{alg:Submatrix}Submatrix $G_{\mathbf{d}_{i}}$}

Initialize: $T_{r}^{(0)}=T_{c}^{(0)}=G_{i},$ $t=0.$
\begin{enumerate}
\item Look for a submatrix $T_{r}^{(t+1)}$ of $G_{i}$ that can be obtained
by removing a row from $T_{r}^{(t)}$ and has the property that \[
\mbox{rank}(T_{r}^{(t+1)}(W,U))\geq\sum_{j\in U}\ell_{i}(j)+\sum_{k\in W}\ell_{i+1}(k)-R\]
 for all $U\subseteq\left[m_{i}\right]$ and $W\subseteq\left[m_{i+1}\right]$.
If such a $T_{r}^{(t+1)}$ exists, then $t\leftarrow t+1$ and repeat
1. Otherwise $T_{c}^{(0)}\leftarrow T_{r}^{(t)},t\leftarrow0,$ and
goto 2.
\item Look for a submatrix $T_{c}^{(t+1)}$ of $G_{i}$ that can be obtained
by removing a column from $T_{c}^{(t)}$ and has the property that
\[
\mbox{rank}(T_{c}^{(t+1)}(W,U))\geq\sum_{j\in U}\ell_{i}(j)+\sum_{k\in W}\ell_{i+1}(k)-R\]
 for all $U\subseteq\left[m_{i}\right]$ and $W\subseteq\left[m_{i+1}\right]$.
If such a $T_{c}^{(t+1)}$ exists, then $t\leftarrow t+1$ and repeat
2. Otherwise output $G_{\mathbf{d}_{i}}=T_{c}^{(t)}$.
\end{enumerate}

\end{algorithm}

Suppose that matrix $G_{i}$ and flow $\mathbf{d}_{i}$ satisfy the
conditions of Theorem \ref{thm:matrix-transversal}. We first prove
that Algorithm \ref{alg:Submatrix} will find a solution $G_{\mathbf{d}_{i}}$
for flow $\mathbf{d}_{i}$, and we later discuss the complexity of
the algorithm. We begin by showing that after exiting Step 1, matrix
$T_{c}^{(0)}$ has exactly $\sum_{j=1}^{m_{i+1}}\ell_{i+1}(j)=R$
rows. By definition, $\mbox{rank}(T_{c}^{(0)}(\left[m_{i+1}\right],\left[m_{i}\right]))\geq\sum_{j}\ell_{i}(j)+\sum_{k}\ell_{i+1}(k)-R=R.$
Suppose that $T_{c}^{(0)}$ has more than $R$ rows. By Theorem \ref{thm:matrix-transversal},
$T_{c}^{(0)}$ has a submatrix $T_{\mathbf{d}_{i}}$ which is a solution
for flow $\mathbf{d}_{i},$ and hence $T_{\mathbf{d}_{i}}$ has $R$
rows. Therefore all rows that belong to $T_{c}^{(0)}$ but not to
$T_{\mathbf{d}_{i}}$ can be removed and the resulting matrix $T$
satisfies $\mbox{rank}(T(W,U))\geq\sum_{j\in U}\ell_{i}(j)+\sum_{k\in W}\ell_{i+1}(k)-R$
for all $U\subseteq\left[m_{i}\right]$ and $W\subseteq\left[m_{i+1}\right]$.
This contradicts the assumption that $T_{c}^{(0)}$ has no more rows
that can be removed. We can similarly argue that the matrix $G_{\mathbf{d}_{i}}$
output at the end of Step 2 has $R$ columns. Since $G_{\mathbf{d}_{i}}$
is an $R\times R$ matrix for which $\mbox{rank}(G_{\mathbf{d}_{i}}(W,U))\geq\sum_{j\in U}\ell_{i}(j)+\sum_{k\in W}\ell_{i+1}(k)-R$
for all $U\subseteq\left[m_{i}\right]$ and $W\subseteq\left[m_{i+1}\right]$,
$G_{\mathbf{d}_{i}}$ is a solution for flow $\mathbf{d}_{i}.$

To find the complexity of Algorithm \ref{alg:Submatrix}, we argue
that at Step 1 or Step 2 a removable row or column can be found respectively
by testing at most $R+1$ rows or columns. For example, in Step 1,
any row which is not part of $T_{c}^{(0)}$ can be removed without
violating any of the stated rank conditions. The pigeonhole principle
implies that at most $R+1$ rows need to be checked to find a row
which is not part of $T_{c}^{(0)}$. Define $\alpha(\Omega)=\mbox{rank}(G_{i}(\Omega))-\sum_{\mathcal{O}_{i}(j)\in\Omega\cap\mathcal{O}_{i}}\ell_{i}(j)+\sum_{\mathcal{O}_{i+1}(k)\in\Omega\cap\mathcal{O}_{i+1}}\ell_{i+1}(k)$
for every cut $\Omega\subseteq\mathcal{O}_{i}\cup\mathcal{O}_{i+1}$.
Then the conditions of Theorem \ref{thm:matrix-transversal} are equivalent
to the requirement that $\alpha(\Omega)\geq0$ for every cut $\Omega\subseteq\mathcal{O}_{i}\cup\mathcal{O}_{i+1}$.
Lemma \ref{lem:9} in the appendix proves the submodularity of $\mbox{rank}(G_{i}(\Omega)).$
By Lemma \ref{lem:91}, which appears later in the appendix, the function
$-\sum_{\mathcal{O}_{i}(j)\in\Omega\cap\mathcal{O}_{i}}\ell_{i}(j)+\sum_{\mathcal{O}_{i+1}(k)\in\Omega\cap\mathcal{O}_{i+1}}\ell_{i+1}(k)$
is also a submodular function of $\Omega.$ Thus $\alpha(\Omega)$,
which is the sum of two submodular functions, is also submodular.
We next verify whether or not the minimum of $\alpha(\Omega)$ is
non--negative. It is known (see, e.g., \cite{key-11} and \cite{key-18})
that the minimum value of a submodular function can be found in strongly
polynomial time. Here we use the algorithm by Schrijver \cite{key-11}
that finds the minimum of a submodular function $f$ on the power
set of set $E_{f}$, in time $O(|E_{f}|^{6}t_{f})$, where $t_{f}$
is the time for evaluating function $f$ for some subset of $E_{f}$.
In our problem $f$ is $\alpha$ which is defined on $E_{\alpha}=\mathcal{O}_{i}\cup\mathcal{O}_{i+1}.$
If we suppose that $|\mathcal{O}_{i}|\leq m$, then $|E_{\alpha}|\leq2m$.
Every evaluation of $\alpha$ requires calculating the rank of submatrix
$G_{i}(\Omega)$ with size at most the size of $G_{i}.$ Suppose that
$G_{i}$ has dimension at most $h_{0}\times h_{0}$. The rank of $G_{i}(\Omega)$
can then be evaluated, for instance, by Gaussian elimination in $O\left(h_{0}^{3}\right)$
time \cite{key-3}. The complexity of minimizing $\alpha(\Omega)$
is therefore \[
O(|E_{\alpha}|^{6}t_{\alpha})=O\left(m^{6}h_{0}^{3}\right).\]
The number of iterations of Step 1 of Algorithm \ref{alg:Submatrix}
is at most the number of rows, $h_{0}.$ Likewise there are at most
$h_{0}$ iterations of Step 2. Since each iteration requires at most
$R+1$ minimization of the submodular function $\alpha,$ the total
complexity of the algorithm is at most \[
(2h_{0})\cdot(R+1)\cdot O\left(m^{6}h_{0}^{3}\right)=O\left(Rm^{6}h_{0}^{4}\right).\]
 As we run Algorithm \ref{alg:Submatrix} for $G_{i}$, $i\in\left[M-1\right],$
the complexity of this part will be $O\left(MRm^{6}h_{0}^{4}\right).$

\subsection*{Proof of Theorem \ref{thm:alg-network}: }

First we prove that the capacity $\mathcal{C}$ of the network $\mathcal{N}$
can be computed in polynomial time. For any cut $\Omega\subseteq\mathcal{V},$
$\mathcal{C}(\Omega)=\sum_{i=1}^{M-1}\mbox{rank}(G_{i}(\Omega)).$
By Lemma \ref{lem:9}, $\mbox{rank}(G_{i}(\Omega))$ is a submodular
function of $\Omega$. Thus $\mathcal{C}(\Omega)$ is a submodular
function of $\Omega$. Next suppose that $\Omega_{1}$ and $\Omega_{2}$
are two cuts that separate $\mathcal{S}$ from $\mathcal{D}.$ Then
$\Omega_{1}\cap\Omega_{2}$ and $\Omega_{1}\cup\Omega_{2}$ separate
$\mathcal{S}$ from $\mathcal{D}$. Therefore $\mathcal{C}(\Omega)$
is a submodular function over all cuts $\Omega\subseteq\mathcal{V}$
that separate $\mathcal{S}$ from $\mathcal{D}.$ To evaluate $\mathcal{C}=\min\mathcal{C}(\Omega)$
over all cuts that separate $\mathcal{S}$ from $\mathcal{D}$ we
use a submodular minimization algorithm \cite{key-11} with running
time of $O(|E_{\mathcal{C}}|^{6}t_{\mathcal{C}})$. Here $E_{\mathcal{C}}=\mathcal{V}\backslash\left\{ \mathcal{S},\mathcal{D}\right\} $
is the ground set of $\mathcal{C}(\Omega)$, and $|E_{\mathcal{C}}|\leq m(M-2)$,
where $m$ is the maximum number of nodes in each layer. $t_{\mathcal{C}}$
denotes the time for the evaluation of $\mathcal{C}(\Omega)=\sum_{i=1}^{M-1}\mbox{rank}(G_{i}(\Omega))$
for a given $\Omega,$ and it involves $M-1$ rank evaluations. Recall
that each matrix $G_{i}$ has size at most $h_{0}\times h_{0}$. Then
$G_{i}(\Omega)$ has size at most $h_{0}\times h_{0}$. Therefore
using Gaussian elimination for rank evaluation, we have $t_{\mathcal{C}}=O(Mh_{0}^{3}).$
Therefore, the total complexity of computing the capacity $\mathcal{C}$
is $O(m^{6}M^{7}h_{0}^{3}).$

Next we discuss the complexity of the evaluation of the vectors $\mathbf{d}_{i}.$
As we discussed in the proof of Theorem \ref{thm:extension}, in order
to determine $\ell_{K}(j)$ for a fixed $K$, we need to solve the
optimization problem (\ref{eq:35}) for the vector $\mathbf{x}$.
This is an integer programming problem over the intersection of two
polymatroids. Let $E_{x}=\left\{ x(1),\cdots,x(m_{K})\right\} .$
Then $|E_{x}|\leq m$. Let $t_{x}$ be the time needed for one evaluation
of functions $f_{A}(T)$ and $f_{B}(T)$ as defined in Section 3.
By applying the result of \cite[Theorem 47.1]{key-9}, it follows
that $\mathbf{x}$ can be found in time $O(|E_{x}|^{6}t_{x})$ if
we use the algorithm of \cite{key-11} for minimizing an integer programming
problem over a polymatroid. Recall that\[
f_{A}(T)=\min\left\{ \mathcal{C}(\Omega_{A})-\sum_{\mathcal{O}_{1}(j)\in\Omega_{A}}\ell_{1}(j)+R:\mathcal{O}_{K}\cap\bar{\Omega}_{A}=T\right\} .\]
 Define for cut $\Omega_{A}$ in $\mathcal{N}_{A}$ with $\mathcal{O}_{K}\cap\bar{\Omega}_{A}=T$,
$\phi_{T}(\Omega_{A})=\mathcal{C}(\Omega_{A})-\sum_{\mathcal{O}_{1}(j)\in\Omega_{A}}\ell_{1}(j)+R$.
By Lemmas \ref{lem:90} and \ref{lem:91} the functions $\mathcal{C}(\Omega)$
and $-\sum_{\mathcal{O}_{1}(j)\in\Omega}\ell_{1}(j)+R$ are submodular
over the set of all cuts in $\mathcal{N}_{A}.$ Next, if for two cuts
$\Omega_{1}$ and $\Omega_{2}$ we have $\mathcal{O}_{K}\cap\bar{\Omega}_{1}=T$
and $\mathcal{O}_{K}\cap\bar{\Omega}_{2}=T$ then De Morgan's laws
imply that $\mathcal{O}_{K}\cap(\overline{\Omega_{1}\cap\Omega_{2}})=T$
and $\mathcal{O}_{K}\cap(\overline{\Omega_{1}\cup\Omega_{2}})=T$.
Thus $\phi_{T}(\Omega_{A})$ is a submodular function over all cuts
$\Omega_{A}$ with $\mathcal{O}_{K}\cap\bar{\Omega}_{A}=T.$ Hence
the evaluation of $f_{A}(T)$ involves the minimization of a submodular
function. The complexity of such a minimization over a set $E_{\phi_{T}}$
is $O(|E_{\phi_{T}}|^{6}t_{\phi_{T}}).$ Since the nodes in $\mathcal{O}_{K}\cap\Omega_{A}=T$
are already fixed, the set $E_{\phi_{T}}$ is the set of all nodes
in $\mathcal{O}_{1}\cup\cdots\cup\mathcal{O}_{K-1}.$ Therefore $|E_{\phi_{T}}|\leq m(K-1)$.
$t_{\phi_{T}}$ is the time for an evaluation of $\mathcal{C}(\Omega_{A})-\sum_{\mathcal{O}_{1}(j)\in\Omega_{A}}\ell_{1}(j)+R$
for a given cut $\Omega_{A}.$ Since $\mathcal{C}(\Omega_{A})=\sum_{i=1}^{K-1}\mbox{rank}(G_{i}(\Omega_{A}))$,
there are $K-1$ rank evaluations. Each rank function can be evaluated
in time at most $O(h_{0}^{3})$ using Gaussian elimination. Thus $t_{\phi_{T}}=O(Kh_{0}^{3})$
and $f_{A}(T)$ can be evaluated in time at most\[
O(|E_{\phi_{T}}|^{6}t_{\phi_{T}})=O\left(Kh_{0}^{3}\cdot(m(K-1))^{6}\right)=O\left(K^{7}m^{6}h_{0}^{3}\right).\]
 Similarly, $f_{B}(T)$ can be evaluated in $O\left((M-K)^{7}m^{6}h_{0}^{3}\right).$
The total time for evaluating $f_{A}$ and $f_{B}$ is therefore $O\left(((M-K)^{7}+K^{7})m^{6}h_{0}^{3}\right).$
Therefore evaluating vector $\mathbf{x}$ needs time \[
O\left(m^{6}\cdot((M-K)^{7}+K^{7})m^{6}h_{0}^{3}\right)=O\left(((M-K)^{7}+K^{7})m^{12}h_{0}^{3}\right)\]
The function above is maximized when $K=0$ and the time complexity
is $O\left(M^{7}m^{12}h_{0}^{3}\right).$ Since the vector $\mathbf{x}$
needs to be computed for every layer $i\in\left\{ 2,\cdots,M-1\right\} $,
we find that the total complexity of the second stage is $O\left(M^{8}m^{12}h_{0}^{3}\right).$
Thus the total complexity of constructing the transmission scheme
is \[
O\left(M^{8}m^{12}h_{0}^{3}\right)+O\left(MRm^{6}h_{0}^{4}\right).\]

\subsection*{Proof of Lemma \ref{lem:tight-submatrix}:}

Let $T^{(0)}=G_{i}$ and for $t\geq1$ define $T^{(t)}$ to be a submatrix
of $G_{i}$ obtained by removing an arbitrary row from $T^{(t-1)}.$
Observe that for every $t\geq1,$ $\hat{U}\subseteq\left[m_{i}\right],$
and $\hat{W}\subseteq\left[m_{i+1}\right],$\begin{equation}
\mbox{rank}(T^{(t)}(\hat{W},\hat{U}))\geq\mbox{rank}(T^{(t-1)}(\hat{W},\hat{U}))-1.\label{eq:6B}\end{equation}
 Suppose that (\ref{eq:transversal-condition}) is satisfied with
strict inequality for all $U\subseteq\left[m_{i}\right]$ and $W\subseteq\left[m_{i+1}\right].$
Then by (\ref{eq:6B}), for all $\hat{U}\subseteq\left[m_{i}\right]$
and $\hat{W}\subseteq\left[m_{i+1}\right],$\[
\mbox{rank}(T^{(1)}(\hat{W},\hat{U}))\geq\sum_{j\in\hat{U}}\ell_{i}(j)+\sum_{k\in\hat{W}}\ell_{i+1}(k)-R_{\mathbf{d}_{i}}.\]
 If $T^{(1)}$ has a tight submatrix we are done. Otherwise, let $T^{(\lambda)}$
denote the empty submatrix of $G_{i}$. Observe that\[
\mbox{rank}(T^{(\lambda)}(\left[m_{i+1}\right],\left[m_{i}\right]))=0<\sum_{j=1}^{m_{i}}\ell_{i}(j)+\sum_{k=1}^{m_{i+1}}\ell_{i+1}(k)-R_{\mathbf{d}_{i}}=R_{\mathbf{d}_{i}},\]
 and so the rows of $G_{i}$ cannot be removed indefinitely without
violating at least one rank condition. Therefore by (\ref{eq:6B})
there must be some $t<\lambda$ such that $T=T^{(t)}$ satisfies (\ref{eq:6A})
for all $\hat{U}$ and $\hat{W}$ and $T$ has a tight submatrix.

\subsection*{Proof of Lemma \ref{lem:proper-submatrix}:}

Consider a tight submatrix $G_{i}(W,U)$ of $G_{i}$ that is not proper.
Suppose, without loss of generality, that $G_{i}(W,U)$ is of the
form $P=G_{i}(\left\{ 1\right\} ,\left[m_{i}\right]).$ Remove rows
from $G_{i}$ arbitrarily among the row blocks with indices in $\left\{ 2,3,\cdots,m_{i+1}\right\} $,
until no further rows can be removed from the resulting submatrix
$T$ of $G_{i}$ without violating $\mbox{rank}(T(\hat{W},\hat{U}))\geq\sum_{j\in\hat{U}}\ell_{i}(j)+\sum_{k\in\hat{W}}\ell_{i+1}(k)-R_{\mathbf{d}_{i}}$
for some $\hat{U}\subseteq\left[m_{i}\right]$ and $\hat{W}\subseteq\left[m_{i+1}\right]$.
Notice that this process terminates before we remove all rows from
blocks with indices in $\left\{ 2,3,\cdots,m_{i+1}\right\} $ for
if $T$ is a matrix with rows only from row block $1$, then $\mbox{rank}(T(\left\{ 2,\cdots,m_{i+1}\right\} ,\left[m_{i}\right]))=0<\sum_{j=1}^{m_{i}}\ell_{i}(j)+\sum_{k=2}^{m_{i+1}}\ell_{i+1}(k)-R_{\mathbf{d}_{i}}=\sum_{k=2}^{m_{i+1}}\ell_{i+1}(k).$
Next replace $G_{i}$ with $T$ for the rest of our argument. Since
no row can be removed from $G_{i}$ among the row blocks with indices
in $\left\{ 2,3,\cdots,m_{i+1}\right\} $ without violating a rank
condition, there exists a tight submatrix $Q=G_{i}(W',U')$ of $G_{i}$
which has a non--empty intersection with the row blocks of $G_{i}$
with indices in $\left\{ 2,3,\cdots,m_{i+1}\right\} .$ Now consider
some possibilities for $Q.$ If $Q$ is a proper submatrix, then we
are done. If $Q$ is not a proper submatrix consider two cases:
\begin{itemize}
\item $|W'|=1$ and $|U'|=m_{i}.$ Without loss of generality suppose that
$Q=G_{i}(\left\{ 2\right\} ,\left[m_{i}\right]).$ Since $P$ and
$Q$ are tight, (\ref{eq:transversal-condition}) implies $\mbox{rank}(P)=\ell_{i+1}(1)$
and $\mbox{rank}(Q)=\ell_{i+1}(2).$ Next consider the submatrix of
$G_{i}$ given by \[
E=\left[\begin{array}{c}
P\\
Q\end{array}\right]=G_{i}\left(\left\{ 1,2\right\} ,\left[m_{i}\right]\right).\]
By (\ref{eq:transversal-condition}), we have $\mbox{rank}(E)\geq\ell_{i+1}(1)+\ell_{i+1}(2).$
However, by considering the rows of $E,P,$ and $Q,$ we see that,
$\mbox{rank}(E)\leq\mbox{rank}(P)+\mbox{rank}(Q)$, and therefore
$\mbox{rank}(E)\leq\ell_{i+1}(1)+\ell_{i+1}(2)$. Thus, $\mbox{rank}(E)=\ell_{i+1}(1)+\ell_{i+1}(2)$.
$E$ is therefore both a tight and a proper submatrix of $G_{i}$. 
\item $|W'|=m_{i+1}$ and $|U'|=1.$ Without loss of generality suppose
that $Q=G_{i}\left(\left[m_{i+1}\right],\left\{ 1\right\} \right)$.
We will next consider a collection of subcases. Suppose first that
other than $P$ and $Q$ there exists another tight submatrix $K$.
If $K$ is proper then there is nothing further to prove. If $K$
is not proper then without loss of generality we can assume that $K$
is either of the form $K=G_{i}(\left\{ 2\right\} ,\left[m_{i}\right])$
or $K=G_{i}\left(\left[m_{i+1}\right],\left\{ 2\right\} \right).$
If $K=G_{i}(\left\{ 2\right\} ,\left[m_{i}\right])$, then the matrix
$\left[\begin{array}{c}
P\\
K\end{array}\right]$ is proper by our argument in the previous case. Likewise, if $K=G_{i}\left(\left[m_{i+1}\right],\left\{ 2\right\} \right)$
then the matrix $\left[\begin{array}{cc}
Q & K\end{array}\right]$ is proper. Suppose next that $P$ and $Q$ are the only tight submatrices
of $G_{i}.$ We already have assumed that no other row from $G_{i}(\left\{ 2,3,\cdots,m_{i+1}\right\} ,\left[m_{i}\right])=T(\left\{ 2,3,\cdots,m_{i+1}\right\} ,\left[m_{i}\right])$
can be removed without violating some rank condition (\ref{eq:transversal-condition}).
Since $P$ and $Q$ are the only tight submatrices of $G_{i},$ removing
one row from $G_{i}(\left\{ 2,3,\cdots,m_{i+1}\right\} ,\left[m_{i}\right])$
will cause the violation of a rank condition only for the tight submatrix
$Q=G_{i}\left(\left[m_{i+1}\right],\left\{ 1\right\} \right).$ If
a row in $G_{i}\left(\left\{ 2,3,\cdots,m_{i+1}\right\} ,\left\{ 1\right\} \right)$
was a linear combination of some other rows in $Q$ it could be removed
without violating any rank condition for submatrix $Q.$ Since $Q$
is a tight submatrix and we constructed $T$ so that no further rows
could be removed from it without violating a rank condition, this
is impossible. Hence every row in $G_{i}(\left\{ 2,3,\cdots,m_{i+1}\right\} ,\left\{ 1\right\} )$
is independent from all other rows in $Q.$ Thus \begin{equation}
\mbox{rank}(Q)=\mbox{rank}(G_{i}(1,1))+r{}_{i+1}(2)+\cdots+r{}_{i+1}(m_{i+1}),\label{eq:5}\end{equation}
where $r{}_{i+1}(j)$ is the number of rows of the $j$th row block
of $G_{i}$. By inequality (\ref{eq:transversal-condition}) for $U=\left[m_{i}\right]$
and $W=\left\{ j\right\} ,$ we have $\mbox{rank}(G_{i}(\left\{ j\right\} ,\left[m_{i}\right]))\geq\ell_{i+1}(j)$
for every $j\in\left\{ 2,\cdots,m_{i+1}\right\} .$ Furthermore, the
rank of a matrix is at most the number of rows of the matrix, and
hence $r{}_{i+1}(j)\geq\mbox{rank}(G_{i}(\left\{ j\right\} ,\left[m_{i}\right]))$.
Consequently $r{}_{i+1}(j)\geq\ell{}_{i+1}(j).$ This together with
(\ref{eq:5}) and the fact that $\mbox{rank}(Q)=\ell_{i}(1)$ imply
that\begin{equation}
\mbox{rank}(G_{i}(1,1))\leq\ell_{i}(1)-\ell_{i+1}(2)-\cdots-\ell_{i+1}(m_{i+1})=\ell_{i}(1)+\ell_{i+1}(1)-R_{\mathbf{d}_{i}}.\label{eq:6}\end{equation}
However, by evaluating inequality (\ref{eq:transversal-condition})
for $U=\left\{ 1\right\} $ and $W=\left\{ 1\right\} $ we obtain\begin{equation}
\mbox{rank}(G_{i}(1,1))\geq\ell_{i}(1)+\ell_{i+1}(1)-R_{\mathbf{d}_{i}}.\label{eq:66}\end{equation}
 (\ref{eq:6}) and (\ref{eq:66}) imply that\[
\mbox{rank}(G_{i}(1,1))=\ell_{i}(1)+\ell_{i+1}(1)-R_{\mathbf{d}_{i}}.\]
 Hence $G_{i}(1,1)$ is tight, and this contradicts the assumption
that $P$ and $Q$ are the only tight submatrices of $G_{i}$. 
\end{itemize}

\subsection*{Proof of Lemma \ref{lem:3}: }

Let $p$ be the index of the first column block of $G_{A}$ and let
$n_{i}+1,\cdots,m_{i}$ respectively be the indices of the other column
blocks. Furthermore let $1,\cdots,n_{i+1}$ be the indices of the
row blocks of $G_{A}.$ Select any two subsets $W\subseteq\left[n_{i+1}\right]$
and $U\subseteq\left\{ p,n_{i}+1,\cdots,m_{i}\right\} .$ We consider
two cases: 
\begin{enumerate}
\item If $p\in U:$ consider the submatrix $G_{A}(W,U)$. This submatrix
is the same as $G_{i}(W,U')$ where $U'=(U\backslash\left\{ p\right\} )\cup\left[n_{i}\right].$
Since $G_{i}(W,U')$ satisfies the condition of inequality (\ref{eq:transversal-condition})
for vector $\mathbf{d}_{i},$ we have:\begin{equation}
\mbox{rank}(G_{A}(W,U))\geq\sum_{j\in W}\ell_{i+1}(j)+\sum_{k\in U'}\ell_{i}(k)-R_{\mathbf{d}_{i}}.\label{eq:7}\end{equation}
 If we expand the right hand side of (\ref{eq:7}), we have:\begin{align}
\mbox{rank}(G_{A}(W,U))\geq & \sum_{j\in W}\ell_{i+1}(j)+\sum_{k\in U\backslash\left\{ p\right\} }\ell_{i}(k)+\sum_{t=1}^{n_{i}}\ell_{i}(t)-R_{\mathbf{d}_{i}}\nonumber \\
= & \sum_{j\in W}\ell_{i+1}(j)+\sum_{k\in U\backslash\left\{ p\right\} }\ell_{i}(k)-\sum_{t=n_{i}+1}^{m_{i}}\ell_{i}(t)\nonumber \\
= & \sum_{j\in W}\ell_{i+1}(j)+\sum_{k\in U\backslash\left\{ p\right\} }\ell_{i}(k)-(R_{A}-\mbox{rank}(P))\label{eq:17D}\end{align}
where (\ref{eq:17D}) follows by (\ref{eq:rank(P)}). The last expression
is the condition of inequality (\ref{eq:transversal-condition}) for
$G_{A}(W,U)$ when $G_{A}$ supports flow $\mathbf{d}_{A}.$
\item If $p\notin U:$ define $V=U\cup\left\{ p\right\} $ so that we can
apply (\ref{eq:17D}): \begin{equation}
\mbox{rank}(G_{A}(W,V))\geq\sum_{j\in W}\ell_{i+1}(j)+\sum_{k\in U}\ell_{i}(k)+\mbox{rank}(P)-R_{A}.\label{eq:8}\end{equation}
 Observe that \begin{align}
\mbox{rank}(G_{A}(W,U))\geq & \mbox{ rank}(G_{A}(W,V))-\mbox{rank}(G_{A}(W,\left\{ p\right\} ))\nonumber \\
\geq & \mbox{ rank}(G_{A}(W,V))-\mbox{rank}(P).\label{eq:18C}\end{align}
 By (\ref{eq:8}) and (\ref{eq:18C}),\[
\mbox{rank}(G_{A}(W,U))\geq\sum_{j\in W}\ell_{i+1}(j)+\sum_{k\in U}\ell_{i}(k)-R_{A},\]
which is the rank condition (\ref{eq:transversal-condition}) for
$G_{A}(W,U)$ when $G_{A}$ supports flow $\mathbf{d}_{A}.$
\end{enumerate}
The block matrix $G_{A}$ has block dimension $n_{i+1}\times(m_{i}-n_{i}+1).$
Therefore, $G_{A}$ has strictly fewer blocks than $G_{i}$ unless
$n_{i+1}=m_{i+1}$ and $m_{i}-n_{i}+1=m_{i}$ or $n_{i}=1.$ This
is impossible since $P$ is a proper submatrix of $G_{i}.$ It follows
from our induction hypothesis matrix $G_{A}$ supports flow $\mathbf{d}_{A}.$

\subsection*{Proof of Lemma \ref{lem:F_B flow}: }

Let the indices for the row and column blocks of $G_{B}$ respectively
be $p,n_{i+1}+1,\cdots,m_{i+1}$ and $1,\cdots,n_{i}.$ Define $G'_{B}=\left[\begin{array}{c}
P\\
B\end{array}\right]$. Observe that matrix $G'_{B}$ supports flow $\mathbf{d}_{B}$ since
the only change needed to the proof of Lemma \ref{lem:3} is to take
the transposition of all matrices. Next, since $G{}_{\mathbf{d}_{A}}$
is a solution for the flow $\mathbf{d}_{A},$ then by definition of
flow the matrix $P_{t}$ has full column rank and has $\mbox{rank}(P)$
columns. Therefore \begin{equation}
\mbox{rank}(P_{t})=\mbox{rank}(P).\label{eq:28C}\end{equation}
 Since $P_{c}$ is a submatrix of $P$ and $P_{t}$ is a submatrix
of $P_{c}$, we have \[
\mbox{rank}(P)\geq\mbox{rank}(P_{c})\geq\mbox{rank}(P_{t}),\]
and therefore\begin{equation}
\mbox{rank}(P_{c})=\mbox{rank}(P).\label{eq:28A}\end{equation}
This implies that the rows that are in $P$ but not in $P_{c}$ are
linear combination of the rows of $P_{c}.$ It follows that for every
$U\subseteq\left[n_{i}\right]$ the rows that are in $P(\left[n_{i+1}\right],U)$
but not in $P_{c}(\left[n_{i+1}\right],U)$ are linear combination
of the rows of $P_{c}(\left[n_{i+1}\right],U)$. Hence \[
\mbox{rank}(P_{c}(\left[n_{i+1}\right],U))=\mbox{rank}(P(\left[n_{i+1}\right],U)).\]
This implies that for every $W\subseteq\left\{ p,n_{i+1}+1,\cdots,m_{i+1}\right\} ,$
if $p\in W$ then $G_{B}(W,U)=\left[\begin{array}{c}
P_{c}(\left[n_{i+1}\right],U)\\
G_{i}(W\backslash p,U)\end{array}\right]$ and \[
\mbox{rank}(G_{B}(W,U))=\mbox{rank}\left(\left[\begin{array}{c}
P_{c}(\left[n_{i+1}\right],U)\\
G_{i}(W\backslash p,U)\end{array}\right]\right)=\mbox{rank}\left(\left[\begin{array}{c}
P(\left[n_{i+1}\right],U)\\
G_{i}(W\backslash p,U)\end{array}\right]\right)=\mbox{rank}(G'_{B}(W,U)).\]
 If $p\notin W,$ $G_{B}(W,U)$ does not depend on $P_{c}$ and $\mbox{rank}(G_{B}(W,U))=\mbox{rank}(G'_{B}(W,U)).$
Therefore in general for every $W$ and $U,$ replacing $P_{c}$ with
$P$ in $G_{B}$ will not change any rank function and we can still
use the result of Lemma \ref{lem:3}.

\subsection*{Proof of Lemma \ref{lem:F_i flow}: }

We have to verify two properties of $G_{\mathbf{d}_{i}}$. The first
is that every row block $j\in\left[m_{i+1}\right]$ of $G_{\mathbf{d}_{i}}$
has $\ell_{i+1}(j)$ rows and every column block $k\in\left[m_{i}\right]$
has $\ell_{i}(k)$ columns. The second property is that $\mbox{rank}(G_{\mathbf{d}_{i}})=R_{\mathbf{d}_{i}}.$
By our construction, the number of columns in column blocks $j\in\left\{ n_{i}+1,\cdots,m_{i}\right\} $
and the number of rows in row blocks $k\in\left[n_{i+1}\right]$ are
respectively determined by $G_{\mathbf{d}_{A}}.$ Since\[
\mathbf{d}_{A}=\left(\mbox{rank}(P),\ell_{i}(n_{i}+1),\ell_{i}(n_{i}+2),\cdots,\ell_{i}(m_{i});\ell_{i+1}(1),\cdots,\ell_{i+1}(n_{i+1})\right),\]
there are the right number of rows and columns in these cases. Furthermore,
the number of columns in column blocks $j\in\left[n_{i}\right]$ and
the number of rows in row blocks $i\in\left\{ n_{i+1}+1,\cdots,m_{i+1}\right\} $
are respectively determined by $G_{\mathbf{d}_{B}}.$ Since \[
\mathbf{d}_{B}=\left(\ell_{i}(1),\cdots,\ell_{i}(n_{i});\mbox{rank}(P),\ell_{i+1}(n_{i+1}+1),\cdots,\ell_{i+1}(m_{i+1})\right)\]
 there are the right number of columns and rows in these cases as
well. Hence every row block $j\in\left[m_{i+1}\right]$ of $G_{\mathbf{d}_{i}}$
has $\ell_{i+1}(j)$ rows and every column block $k\in\left[m_{i}\right]$
has $\ell_{i}(k)$ columns.

For the second property, observe that since $G_{\mathbf{d}_{B}}$
is a solution of flow $\mathbf{d}_{B}$ for $G_{B},$ then by the
definition of the solution of a flow the matrix $P_{ct}$ has full
row rank and has $\mbox{rank}(P)$ rows. Therefore \[
\mbox{rank}(P_{ct})=\mbox{rank}(P).\]
 Since $P_{ct}$ is a submatrix of $P_{cr}$ consisting of a subset
of its rows, $\mbox{rank}(P_{cr})\geq\mbox{rank}(P_{ct})$. Because
$P_{cr}$ is a submatrix of $P,$ $\mbox{rank}(P_{cr})\leq\mbox{rank}(P).$
Thus \begin{equation}
\mbox{rank}(P_{cr})=\mbox{rank}(P),\label{eq:28B}\end{equation}
 and so the rows that are in $P_{cr}$ are linear combinations of
the rows in $P_{ct}.$ Next since the rows of $B_{t}$ are also rows
of the full--rank matrix $G_{\mathbf{d}_{B}},$ it follows that all
rows of $B_{t}$ are linearly independent and are independent from
all other rows of $G_{\mathbf{d}_{B}}$ and consequently from all
other rows in $G'_{\mathbf{d}_{B}}.$ Therefore, all rows in $\left[\begin{array}{cc}
B_{t} & L_{t}\end{array}\right]$ are linearly independent and are independent from all other rows
in $G_{\mathbf{d}_{i}}.$ Therefore the two relationships follow:
\begin{equation}
\mbox{rank}(G_{\mathbf{d}_{i}})=\mbox{rank}\left(\left[\begin{array}{cc}
B_{t} & L_{t}\end{array}\right]\right)+\mbox{rank}\left(\left[\begin{array}{cc}
P_{cr} & A_{t}\end{array}\right]\right),\label{eq:9}\end{equation}
\begin{equation}
\mbox{rank}\left(\left[\begin{array}{cc}
B_{t} & L_{t}\end{array}\right]\right)=\sum_{j=n_{i+1}+1}^{m_{i+1}}\ell_{i+1}(j).\label{eq:10}\end{equation}
 Notice that by (\ref{eq:28A}) and (\ref{eq:28B}) $\mbox{rank}(P_{cr})=\mbox{rank}(P_{c})$.
Since $P_{cr}$ is a submatrix of $P_{c}$ consisting of a subset
of its columns, the columns of $P_{c}$ are linear combinations of
the columns of $P_{cr}.$ Therefore\begin{equation}
\mbox{rank}\left(\left[\begin{array}{cc}
P_{cr} & A_{t}\end{array}\right]\right)=\mbox{rank}\left(\left[\begin{array}{cc}
P_{c} & A_{t}\end{array}\right]\right).\label{eq:10A}\end{equation}
By (\ref{eq:28C}) and (\ref{eq:28A}), $\mbox{rank}(P_{c})=\mbox{rank}(P_{t}).$
Recall that $P_{t}$ is a submatrix of $P_{c}$ consisting of a subset
of its columns. Therefore the columns of $P_{c}$ are linear combinations
of the columns of $P_{t}$. Hence,\begin{equation}
\mbox{rank}\left(\left[\begin{array}{cc}
P_{c} & A_{t}\end{array}\right]\right)=\mbox{rank}\left(\left[\begin{array}{cc}
P_{t} & A_{t}\end{array}\right]\right).\label{eq:10B}\end{equation}
 Since $G_{\mathbf{d}_{A}}=\left[\begin{array}{cc}
P_{t} & A_{t}\end{array}\right]$ is a solution for flow $\mathbf{d}_{A},$ it follows that $\mbox{rank}\left(G_{\mathbf{d}_{A}}\right)=\sum_{j=1}^{n_{i+1}}\ell_{i+1}(j).$
Thus, (\ref{eq:10A}) and (\ref{eq:10B}) imply \begin{equation}
\mbox{rank}\left(\left[\begin{array}{cc}
P_{cr} & A_{t}\end{array}\right]\right)=\sum_{j=1}^{n_{i+1}}\ell_{i+1}(j).\label{eq:11}\end{equation}
 From (\ref{eq:9}), (\ref{eq:10}), and (\ref{eq:10A}) we conclude
that\[
\mbox{rank}\left(G_{\mathbf{d}_{i}}\right)=\sum_{j=1}^{m_{i+1}}\ell_{i+1}(j)=R_{\mathbf{d}_{i}}.\]

\subsection*{Proof of Lemma \ref{lem:7}:}

\subsubsection*{Part I: }

Here we only prove the submodularity of function $f_{A}(T)$ as the
proof for $f_{B}(T)$ is similar. Consider layers $\mathcal{O}_{i}$
and $\mathcal{O}_{i+1}$ and transfer function $G_{i}.$ Recall that
$G_{i}(\Omega)$ is the transfer function from the nodes in $\Omega$
to the nodes in $\bar{\Omega}=\left(\mathcal{O}_{i}\cup\mathcal{O}_{i+1}\right)\backslash\Omega.$
We first prove that:
\begin{lem}
\label{lem:9} $\mbox{rank}(G_{i}(\Omega))$ is a submodular function
over cuts $\Omega\subseteq\mathcal{O}_{i}\cup\mathcal{O}_{i+1}$;
i.e., for every two such cuts $\Omega_{1}$ and $\Omega_{2}:$ \[
\mbox{rank}\left(G_{i}\left(\Omega_{1}\right)\right)+\mbox{rank}\left(G_{i}\left(\Omega_{2}\right)\right)\geq\mbox{rank}\left(G_{i}\left(\Omega_{1}\cap\Omega_{2}\right)\right)+\mbox{rank}\left(G_{i}\left(\Omega_{1}\cup\Omega_{2}\right)\right).\]

\end{lem}
We point out that the preceding result was first proved in \cite{key-1}
in order to study the time--expanded representation of a network which
is not layered. \cite{key-1} established Lemma \ref{lem:9} through
an information theoretic argument involving the submodularity of the
entropy function. We offer a new and combinatorial proof of Lemma
\ref{lem:9}.
\begin{proof}
Consider matrix $G_{i}$ in Figure \ref{fig:1}-(a). Suppose that
we have reordered the row blocks of $G_{i}$ and the column blocks
of $G_{i}$ such that the blocks corresponding to the transfer matrices
of $G_{i}(\Omega_{1})$ and $G_{i}(\Omega_{2})$ appear as in Figure
\ref{fig:1}-(a). %
\begin{figure}
\includegraphics[bb=5bp 0bp 530bp 250bp,clip,scale=0.6]{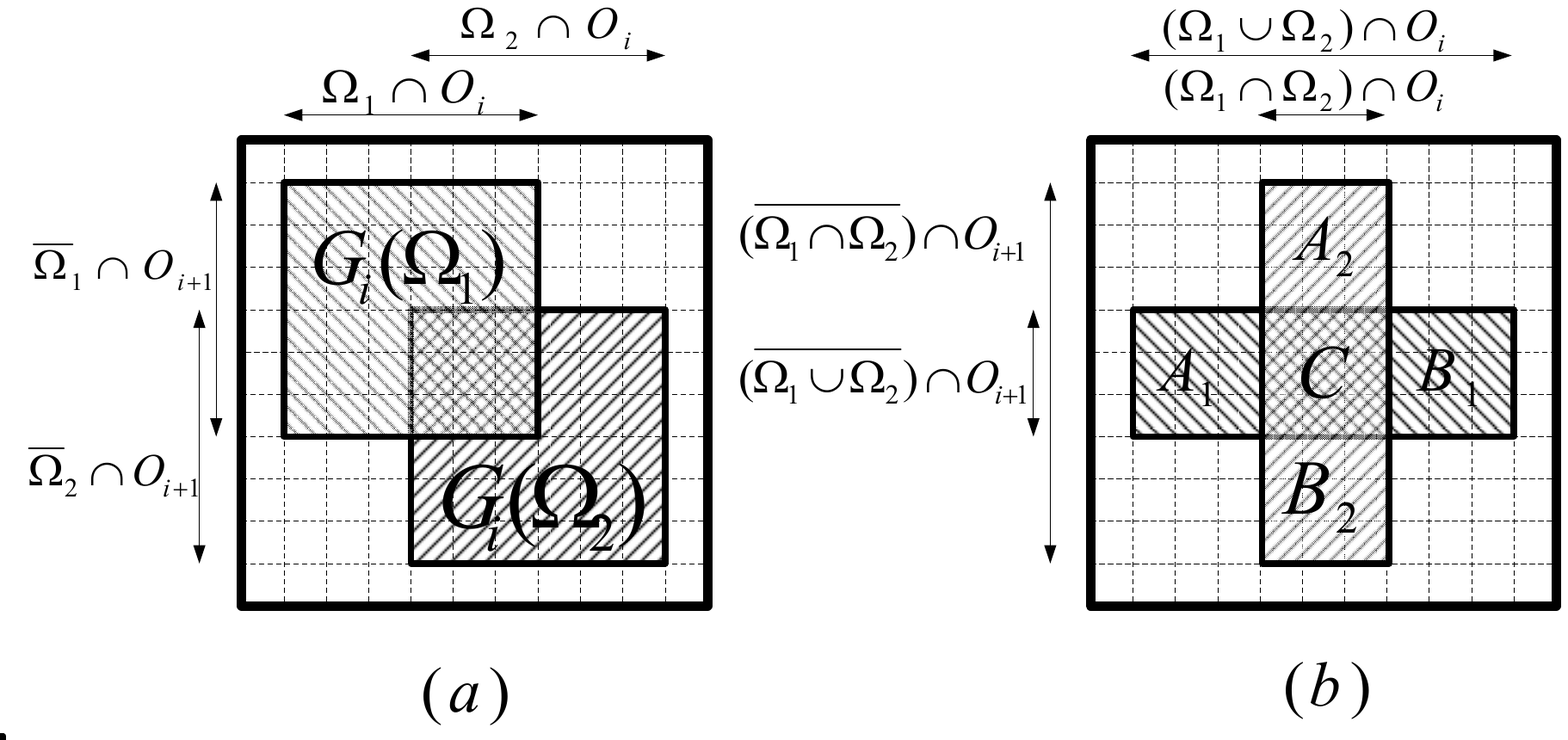}\centering\caption{\label{fig:1}(a) The matrices $G_{i}(\Omega_{1})$ and $G_{i}(\Omega_{2})$,
(b) the matrices $G_{i}(\Omega_{1}\cap\Omega_{2})$ and $G_{i}(\Omega_{1}\cup\Omega_{2})$.}

\end{figure}
We have depicted and labeled the different parts of the transfer matrices
$G_{i}(\Omega_{1}\cap\Omega_{2})$ and $G_{i}(\Omega_{1}\cup\Omega_{2})$
in Figure \ref{fig:1}-(b). Therefore we have:\begin{align*}
G_{i}(\Omega_{1}\cap\Omega_{2})=\left[\begin{array}{c}
A_{2}\\
C\\
B_{2}\end{array}\right],\quad & G_{i}(\Omega_{1}\cup\Omega_{2})=\left[\begin{array}{ccc}
A_{1} & C & B_{1}\end{array}\right].\end{align*}
 We first prove that:\begin{equation}
\mbox{rank}\left(G_{i}(\Omega_{1})\right)+\mbox{rank}\left(C\right)\geq\mbox{rank}\left(\left[\begin{array}{cc}
A_{1} & C\end{array}\right]\right)+\mbox{rank}\left(\left[\begin{array}{c}
A_{2}\\
C\end{array}\right]\right).\label{eq:22}\end{equation}
Let $p$ be the maximum number of rows $a_{1},\cdots,a_{p}$ in $A_{2}$
which are independent in $\left[\begin{array}{c}
a_{1}\\
\vdots\\
a_{p}\\
C\end{array}\right]$ and $q$ be the maximum number of rows $b_{1},\cdots,b_{q}$ in $G_{i}(\Omega_{1})$
but not in $\left[\begin{array}{cc}
A_{1} & C\end{array}\right]$ which are independent in $\left[\begin{array}{c}
b_{1}\\
\vdots\\
b_{q}\\
A_{1}\quad C\end{array}\right].$ We have:\begin{align}
\mbox{rank}\left(\left[\begin{array}{c}
A_{2}\\
C\end{array}\right]\right)= & \mbox{ rank}\left(C\right)+p,\label{eq:22A}\\
\mbox{rank}\left(G_{i}(\Omega_{1})\right)= & \mbox{ rank}\left(\left[\begin{array}{cc}
A_{1} & C\end{array}\right]\right)+q.\label{eq:22B}\end{align}
Let $a'{}_{1},\cdots,a'_{p}$ denote the rows in $G_{i}(\Omega_{1})$
which respectively end in the rows $a{}_{1},\cdots,a{}_{p}$ of $A_{2}.$
Then these rows are clearly independent in $\left[\begin{array}{c}
a'_{1}\\
\vdots\\
a'_{p}\\
A_{1}\quad C\end{array}\right],$ and therefore $q\geq p$. By (\ref{eq:22A}) and (\ref{eq:22B}),\[
\mbox{rank}\left(G_{i}(\Omega_{1})\right)-\mbox{rank}\left(\left[\begin{array}{cc}
A_{1} & C\end{array}\right]\right)\geq\mbox{rank}\left(\left[\begin{array}{c}
A_{2}\\
C\end{array}\right]\right)-\mbox{rank}\left(C\right),\]
 which implies (\ref{eq:22}). A similar argument for $G_{i}(\Omega_{2})$
implies that:\begin{equation}
\mbox{rank}\left(G_{i}(\Omega_{2})\right)+\mbox{rank}(C)\geq\mbox{rank}\left(\left[\begin{array}{cc}
C & B_{1}\end{array}\right]\right)+\mbox{rank}\left(\left[\begin{array}{c}
C\\
B_{2}\end{array}\right]\right).\label{eq:23}\end{equation}
By adding together inequalities (\ref{eq:22}) and (\ref{eq:23})
we find that\begin{eqnarray}
\mbox{rank}\left(G_{i}(\Omega_{1})\right)+\mbox{rank}\left(G_{i}(\Omega_{2})\right) & \geq & \left(\mbox{rank}\left(\left[\begin{array}{cc}
A_{1} & C\end{array}\right]\right)+\mbox{rank}\left(\left[\begin{array}{cc}
C & B_{1}\end{array}\right]\right)-\mbox{rank}(C)\right)\nonumber \\
 & + & \left(\mbox{rank}\left(\left[\begin{array}{c}
A_{2}\\
C\end{array}\right]\right)+\mbox{rank}\left(\left[\begin{array}{c}
C\\
B_{2}\end{array}\right]\right)-\mbox{rank}(C)\right).\label{eq:24}\end{eqnarray}
 If we use the submodularity of the rank function of a matrix \cite{key-5}
we deduce that if $W_{1}$ and $W_{2}$ are the indices of rows (columns)
of some matrix and $\mbox{rank}(W)$ is the number of independent
rows (columns) among those with indices in $W$, then: \begin{equation}
\mbox{rank}(W_{1})+\mbox{rank}(W_{2})\geq\mbox{rank}(W_{1}\cap W_{2})+\mbox{rank}(W_{1}\cup W_{2}).\label{eq:25}\end{equation}
 Applying (\ref{eq:25}) to the columns of matrix $G_{i}(\Omega_{1}\cup\Omega_{2})$
and to the rows of matrix $G_{i}(\Omega_{1}\cap\Omega_{2})$ we find
that: \begin{align}
\mbox{rank}\left(\left[\begin{array}{cc}
A_{1} & C\end{array}\right]\right)+\mbox{rank}\left(\left[\begin{array}{cc}
C & B_{1}\end{array}\right]\right) & \geq\mbox{rank}(C)+\mbox{rank}\left(G_{i}\left(\Omega_{1}\cup\Omega_{2}\right)\right)\label{eq:26}\\
\mbox{rank}\left(\left[\begin{array}{c}
A_{2}\\
C\end{array}\right]\right)+\mbox{rank}\left(\left[\begin{array}{c}
C\\
B_{2}\end{array}\right]\right) & \ge\mbox{rank}(C)+\mbox{rank}\left(G_{i}\left(\Omega_{1}\cap\Omega_{2}\right)\right).\label{eq:27}\end{align}
 Lemma \ref{lem:9} follows from (\ref{eq:24}), (\ref{eq:26}), and
(\ref{eq:27}). 
\end{proof}
Next we extend Lemma \ref{lem:9} to a multilayer network $\mathcal{N}_{A}:$
\begin{lem}
\label{lem:90} In the network $\mathcal{N}_{A}$, $\mathcal{C}(\Omega)$
is a submodular function over cuts $\Omega\subseteq\mathcal{O}_{1}\cup\cdots\cup\mathcal{O}_{K}$.\end{lem}
\begin{proof}
We decompose $\Omega$ into the subsets $\Omega=\Omega_{1}\cup\Omega_{2}\cup\cdots\cup\Omega_{K-1},$
where $\Omega_{i}=\Omega\cap(\mathcal{O}_{i}\cup\mathcal{O}_{i+1})$
defines a cut of the subnetwork of $\mathcal{N}_{A}$ with set of
nodes $\mathcal{O}_{i}\cup\mathcal{O}_{i+1}.$ We have \[
\mathcal{C}(\Omega)=\sum_{j=1}^{K-1}\mbox{rank}\left(G_{j}(\Omega)\right)=\sum_{j=1}^{K-1}\mbox{rank}\left(G_{j}(\Omega_{j})\right).\]
If two cuts $\Omega$ and $\Omega'$ are respectively decomposed into
$\left[\Omega_{i}\right]$ and $\left[\Omega'_{i}\right]$, then $\Omega\cap\Omega'$
and $\Omega\cup\Omega'$ will be respectively decomposed into $\left[\Omega_{i}\cap\Omega'_{i}\right]$
and $\left[\Omega_{i}\cup\Omega'_{i}\right]$. By Lemma \ref{lem:9},
$\mbox{rank}\left(G_{j}(\Omega_{j})\right)$ is submodular. Therefore
\begin{align*}
\mathcal{C}(\Omega)+\mathcal{C}(\Omega')= & \sum_{j=1}^{K-1}\left(\mbox{rank}\left(G_{j}(\Omega_{j})\right)+\mbox{rank}\left(G_{j}(\Omega'_{j})\right)\right)\\
\geq & \sum_{j=1}^{K-1}\left(\mbox{rank}\left(G_{j}(\Omega_{j}\cap\Omega'_{j})\right)+\mbox{rank}\left(G_{j}(\Omega_{j}\cup\Omega'_{j})\right)\right)=\mathcal{C}(\Omega_{j}\cap\Omega'_{j})+\mathcal{C}(\Omega_{j}\cup\Omega'_{j}).\end{align*}
Since the sum of submodular functions is submodular, the final result
follows. 
\end{proof}
We next prove the following useful lemma:
\begin{lem}
\label{lem:91}For any function $f$ defined on the set $E$ and any
given set $V\subseteq E$, the function $g$ on the power set of $E$
defined as $g(U)=\sum_{i\in U\cap V}f(k)$ is submodular.\end{lem}
\begin{proof}
By definition,\begin{align*}
g(U_{1})+g(U_{2})= & \sum_{k\in U_{1}\cap V}f(k)+\sum_{k\in U_{2}\cap V}f(k)\\
= & \sum_{k\in(U_{1}\cap V)\cap(U_{2}\cap V)}f(k)+\sum_{k\in(U_{1}\cap V)\cup(U_{2}\cap V)}f(k)\\
= & \sum_{k\in(U_{1}\cap U_{2})\cap V}f(k)+\sum_{k\in(U_{1}\cup U_{2})\cap V}f(k)\\
= & g(U_{1}\cap U_{2})+g(U_{1}\cup U_{2}).\end{align*}
 
\end{proof}
We next prove that $f_{A}$ is submodular. 
\begin{proof}
Suppose $T_{1},T_{2}\subseteq\mathcal{O}_{K},$ and let $A_{1}$ and
$A_{2}$ be the two cuts in $\mathcal{N}{}_{A}$ such that\begin{align*}
f_{A}(T_{1})=\mathcal{C}(A_{1})-\sum_{\mathcal{O}_{1}(j)\in A_{1}}\ell_{1}(j)+R, & \quad\mathcal{O}_{K}\cap\bar{A}_{1}=T_{1}\\
f_{A}(T_{2})=\mathcal{C}(A_{2})-\sum_{\mathcal{O}_{1}(j)\in A_{2}}\ell_{1}(j)+R, & \quad\mathcal{O}_{K}\cap\bar{A}_{2}=T_{2}\end{align*}
Consider $A_{1}\cap A_{2}$ and $A_{1}\cup A_{2}.$ By De Morgan's
laws, $\mathcal{O}_{K}\cap\left(\overline{A_{1}\cap A_{2}}\right)=\mathcal{O}_{K}\cap\left(\bar{A}_{1}\cup\bar{A}_{2}\right)=T_{1}\cup T_{2}$
and $\mathcal{O}_{K}\cap\left(\overline{A_{1}\cup A_{2}}\right)=\mathcal{O}_{K}\cap\left(\bar{A}_{1}\cap\bar{A}_{2}\right)=T_{1}\cap T_{2}.$
By (\ref{eq:f_A}) we have\begin{align}
f_{A}(T_{1}\cup T_{2})\leq\mathcal{C}(A_{1}\cap A_{2})-\sum_{\mathcal{O}_{1}(j)\in A_{1}\cap A_{2}}\ell_{1}(j)+R,\label{eq:demorgan1}\\
f_{A}(T_{1}\cap T_{2})\leq\mathcal{C}(A_{1}\cup A_{2})-\sum_{\mathcal{O}_{1}(j)\in A_{1}\cup A_{2}}\ell_{1}(j)+R.\label{eq:demorgan2}\end{align}
By the submodularity of $\mathcal{C}(\Omega)$ we have \begin{equation}
\mathcal{C}(A_{1}\cap A_{2})+\mathcal{C}(A_{1}\cup A_{2})\leq\mathcal{C}(A_{1})+\mathcal{C}(A_{2}).\label{eq:cut-sub}\end{equation}
 Furthermore \begin{equation}
\sum_{\mathcal{O}_{1}(j)\in A_{1}\cap A_{2}}\ell_{1}(j)+\sum_{\mathcal{O}_{1}(j)\in A_{1}\cup A_{2}}\ell_{1}(j)=\sum_{\mathcal{O}_{1}(j)\in A_{1}}\ell_{1}(j)+\sum_{\mathcal{O}_{1}(j)\in A_{2}}\ell_{1}(j).\label{eq:set-theory}\end{equation}
By (\ref{eq:demorgan1})--(\ref{eq:set-theory}), we have\[
f_{A}(T_{1}\cap T_{2})+f_{A}(T_{1}\cup T_{2})\leq f_{A}(T_{1})+f_{A}(T_{2}).\]
 
\end{proof}

\subsubsection*{Part II:}

We prove the result for $f_{A}$ and it is straightforward to modify
it for $f_{B}.$ It suffices to prove that for every $T\subseteq\mathcal{O}_{K}$
and every $\mathcal{O}_{K}(i)\notin T,$\[
f_{A}(T)\leq f_{A}\left(T\cup\left\{ \mathcal{O}_{K}(i)\right\} \right).\]
 Suppose that for $T$ and $\mathcal{O}_{K}(i)\notin T,$ cut $A$
in $\mathcal{N}_{A}$ achieves $f_{A}\left(T\cup\left\{ \mathcal{O}_{K}(i)\right\} \right).$
We have\begin{align}
f_{A}\left(T\cup\left\{ \mathcal{O}_{K}(i)\right\} \right)=\mathcal{C}(A)-\sum_{\mathcal{O}_{1}(j)\in A}\ell_{1}(j)+R, & \quad\mathcal{O}_{K}\cap\bar{A}=T\cup\left\{ \mathcal{O}_{K}(i)\right\} .\label{eq:inc1}\end{align}
 Next consider the cut $A'=A\cup\left\{ \mathcal{O}_{K}(i)\right\} .$
Notice that $\bar{A}'\cap\mathcal{O}_{K}=T.$ Therefore \begin{equation}
\mathcal{C}(A')-\sum_{\mathcal{O}_{1}(j)\in A'}\ell_{1}(j)+R\geq f_{A}(T).\label{eq:inc2}\end{equation}
 Observe that for every $j\leq K-2$, $G_{j}(A)=G_{j}(A').$ Furthermore
$G_{K-1}(A)$ has the same row blocks as $G_{K-1}(A')$ and an additional
row block corresponding to the transfer function from the nodes in
$A\cap\mathcal{O}_{K-1}$ to $\left\{ \mathcal{O}_{K}(i)\right\} .$
Therefore $\mathcal{C}(A)\geq\mathcal{C}(A').$ Finally, $\sum_{\mathcal{O}_{1}(j)\in A}\ell_{1}(j)=\sum_{\mathcal{O}_{1}(j)\in A'}\ell_{1}(j)$,
and hence\begin{equation}
\mathcal{C}(A)-\sum_{\mathcal{O}_{1}(j)\in A}\ell_{1}(j)+R\geq\mathcal{C}(A')-\sum_{\mathcal{O}_{1}(j)\in A'}\ell_{1}(j)+R.\label{eq:inc3}\end{equation}
(\ref{eq:inc1}), (\ref{eq:inc2}), and (\ref{eq:inc3}) imply that
$f_{A}(T)\leq f_{A}\left(T\cup\left\{ \mathcal{O}_{K}(i)\right\} \right).$

\subsubsection*{Part III:}

Since $\sum_{\mathcal{O}_{1}(j)\in\Omega_{A}}\ell_{1}(j)\leq R$ and
$\sum_{\mathcal{O}_{M}(k)\in\bar{\Omega}_{B}}\ell_{M}(k)\leq R$,
$f_{A}$ and $f_{B}$ are non--negative functions. Observe that by
choosing cut $A=\mathcal{O}_{1}\cup\cdots\cup\mathcal{O}_{K}$ for
network $\mathcal{N}_{A}$ and cut $B=\emptyset$ for network $\mathcal{N}_{B}$,
we find that $\bar{A}\cap\mathcal{O}_{K}=\emptyset,$ $B\cap\mathcal{O}_{K}=\emptyset,$
and\[
f_{A}(\emptyset)=\mathcal{C}(A)-\sum_{\mathcal{O}_{1}(j)\in A}\ell_{1}(j)+R=f_{B}(\emptyset)=\mathcal{C}(B)-\sum_{\mathcal{O}_{M}(j)\in\bar{B}}\ell_{M}(j)+R=0.\]

\subsection*{Proof of Lemma \ref{lem:final-lemma}:}

Suppose that $T_{0}$ achieves $\min_{T\subseteq\mathcal{O}_{K}}\left(f_{A}(T)+f_{B}(\mathcal{O}_{K}\backslash T)\right)$
and cuts $\Omega_{A}$ and $\Omega_{B}$ respectively achieve $f_{A}(T_{0})$
and $f_{B}(\mathcal{O}_{K}\backslash T_{0})$ in $\mathcal{N}_{A}$
and $\mathcal{N}_{B}.$ Then 

\begin{align}
f_{A}(T_{0})=\mathcal{C}(\Omega_{A})-\sum_{\mathcal{O}_{1}(j)\in\Omega_{A}}\ell_{1}(j)+R, & \quad\mathcal{O}_{K}\cap\bar{\Omega}_{A}=T_{0}\label{eq:A}\\
f_{B}(\mathcal{O}_{K}\backslash T_{0})=\mathcal{C}(\Omega_{B})-\sum_{\mathcal{O}_{M}(j)\in\bar{\Omega}_{B}}\ell_{M}(j)+R, & \quad\mathcal{O}_{K}\cap\Omega_{B}=\mathcal{O}_{K}\backslash T_{0}.\label{eq:B}\end{align}
 Let $\Omega=\Omega_{A}\cup\Omega_{B}$ be a cut in network $\mathcal{N}$.
Since $\mathcal{O}_{K}\cap\Omega_{A}=\mathcal{O}_{K}\cap\Omega_{B}=\mathcal{O}_{K}\backslash T_{0},$
it follows that $G_{i}(\Omega)=G_{i}(\Omega_{A})$ for $i\in\left\{ 1,\cdots,K-1\right\} $
and $G_{i}(\Omega)=G_{i}(\Omega_{B})$ for $i\in\left\{ K,\cdots,M-1\right\} .$
Thus we have \[
\mathcal{C}(\Omega)=\mathcal{C}(\Omega_{A})+\mathcal{C}(\Omega_{B}).\]
 From (\ref{eq:A}) and (\ref{eq:B}) we have\begin{equation}
f_{A}(T_{0})+f_{B}(\mathcal{O}_{K}\backslash T_{0})=\mathcal{C}(\Omega)-\sum_{\mathcal{O}_{1}(j)\in\Omega_{A}}\ell_{1}(j)-\sum_{\mathcal{O}_{M}(j)\in\bar{\Omega}_{B}}\ell_{M}(j)+2R.\label{eq:aboce}\end{equation}
 Since $\sum_{\mathcal{O}_{1}(j)\in\Omega_{A}}\ell_{1}(j)=\sum_{\mathcal{O}_{1}(j)\in\Omega}\ell_{1}(j)$
and $\sum_{\mathcal{O}_{1}(j)\in\bar{\Omega}_{B}}\ell_{M}(j)=\sum_{\mathcal{O}_{M}(j)\in\bar{\Omega}}\ell_{M}(j),$
(\ref{eq:aboce}) and (\ref{eq:necessity}) together imply $f_{A}(T_{0})+f_{B}(\mathcal{O}_{K}\backslash T_{0})\geq R$,
as desired. 

To prove the converse, consider a cut $\Omega$ in network $\mathcal{N}$
and partition it into two cuts $\Omega_{A}$ and $\Omega_{B}$ in
networks $\mathcal{N}_{A}$ and $\mathcal{N}_{B}$ respectively. Let
$T_{0}=\bar{\Omega}\cap\mathcal{O}_{K}.$ Then by definition\begin{align*}
f_{A}(T_{0})\leq\mathcal{C}(\Omega_{A})-\sum_{\mathcal{O}_{1}(j)\in\Omega_{A}}\ell_{1}(j)+R, & \quad\mathcal{O}_{K}\cap\bar{\Omega}_{A}=T_{0}\\
f_{B}(\mathcal{O}_{K}\backslash T_{0})\leq\mathcal{C}(\Omega_{B})-\sum_{\mathcal{O}_{M}(j)\in\bar{\Omega}_{B}}\ell_{M}(j)+R, & \quad\mathcal{O}_{K}\cap\Omega_{B}=\mathcal{O}_{K}\backslash T_{0}.\end{align*}
 Therefore\begin{align*}
R\leq & \min_{T\subseteq\mathcal{O}_{K}}\left(f_{A}(T)+f_{B}(\mathcal{O}_{K}\backslash T)\right)\leq f_{A}(T_{0})+f_{B}(\mathcal{O}_{K}\backslash T_{0})\\
\leq & \mathcal{C}(\Omega_{A})+\mathcal{C}(\Omega_{B})-\sum_{\mathcal{O}_{1}(j)\in\Omega_{A}}\ell_{1}(j)-\sum_{\mathcal{O}_{M}(j)\in\bar{\Omega}_{B}}\ell_{M}(j)+2R.\end{align*}
 If we substitute $\mathcal{C}(\Omega_{A})+\mathcal{C}(\Omega_{B})=\mathcal{C}(\Omega)$
and $\sum_{\mathcal{O}_{1}(j)\in\Omega_{A}}\ell_{1}(j)=\sum_{\mathcal{O}_{1}(j)\in\Omega}\ell_{1}(j)$
and $\sum_{\mathcal{O}_{M}(j)\in\bar{\Omega}_{B}}\ell_{M}(j)=\sum_{\mathcal{O}_{M}(j)\in\bar{\Omega}}\ell_{M}(j)$
into the preceding expression, we obtain \[
\mathcal{C}(\Omega)\geq\sum_{\mathcal{O}_{1}(j)\in\Omega}\ell_{1}(j)+\sum_{\mathcal{O}_{1}(j)\in\bar{\Omega}}\ell_{M}(j)-R,\]
 which is the final result.
\end{document}